\documentclass[sigconf]{acmart}

\usepackage{amsfonts} 
\usepackage{verbatim}
\usepackage{mathtools}

\AtBeginDocument{%
  \providecommand\BibTeX{{%
    \normalfont B\kern-0.5em{\scshape i\kern-0.25em b}\kern-0.8em\TeX}}}

\newcommand{\E}{\mathbb E}
\newcommand{\tsc}[1]{\textsuperscript{#1}} %shorthand for superscripts

\DeclareMathOperator*{\argmin}{arg\,min}

\copyrightyear{2023} 
\acmYear{2023} 
\setcopyright{acmlicensed}\acmConference[WWW '23 Companion]{Companion Proceedings of the ACM Web Conference 2023}{April 30-May 4, 2023}{Austin, TX, USA}
\acmBooktitle{Companion Proceedings of the ACM Web Conference 2023 (WWW '23 Companion), April 30-May 4, 2023, Austin, TX, USA}
\acmPrice{15.00}
\acmDOI{10.1145/3543873.3584635}
\acmISBN{978-1-4503-9419-2/23/04}

\begin{document}

\title{Anytime-Valid Confidence Sequences in an Enterprise A/B Testing Platform}

\settopmatter{printacmref=false}

%\settopmatter{authorsperrow=1} %make the template consider one author per row
\author{Akash V. Maharaj\tsc{1}, Ritwik Sinha\tsc{2}, David Arbour\tsc{2}, Ian Waudby-Smith\tsc{3}, Simon Z. Liu\tsc{1}, }
\author{Moumita Sinha\tsc{1}, Raghavendra Addanki\tsc{2}, Aaditya Ramdas\tsc{3}, }
\author{Manas Garg\tsc{1}, Viswanathan Swaminathan\tsc{2}}
%to break line, start another author block
\affiliation{
% \institution{\vskip A.1cm}%add some spacing if needed
 \institution{\tsc{1} Adobe Inc., San Jose, USA.}
  \institution{\tsc{2} Adobe Research, San Jose, USA.}
  \institution{\tsc{3} Carnegie Mellon University, Pittsburgh, USA.}
 \country{}
}



\renewcommand{\shortauthors}{Maharaj et al.}

% NOTE: \acmVersion should be set to 1 for the ACM version of the paper, while it should be set to 0 for the arxiv version.
\newcommand{\acmVersion}{0}

\newcommand{\appendixref}[1]{%
    \ifnum\acmVersion=1{#1 of our online preprint \citep{LongerPaper}}%
    \else{#1}%
    \fi%
}

\begin{abstract}
A/B tests are the gold standard for evaluating digital experiences on the web. However, traditional ``fixed-horizon'' statistical methods are often incompatible with the needs of modern industry practitioners as they do not permit continuous monitoring of experiments. Frequent evaluation of fixed-horizon tests (``peeking'') leads to inflated type-I error and can result in erroneous conclusions. 
We have released an experimentation service on the Adobe Experience Platform based on \emph{anytime-valid confidence sequences}, allowing for continuous monitoring of the A/B test and data-dependent stopping. We demonstrate how we adapted and deployed asymptotic confidence sequences in a full featured A/B testing platform, describe how sample size calculations can be performed, and how alternate test statistics like ``lift'' can be analyzed. On both simulated data and thousands of real experiments, we show the desirable properties of using anytime-valid methods instead of traditional approaches.

\end{abstract}

\maketitle

\section{Introduction}

A/B testing (also referred to as randomized experimentation) has become ubiquitous in the optimization of digital experiences on websites, mobile applications, and in emails \cite{kohavi2020trustworthy}. The rapid growth and commoditization of experimentation platforms has made it easier than ever for developers, product managers, marketers, analysts, designers, and business leaders to use A/B testing tools. While this ubiquity has led to the growth of data-driven decision-making within organizations, it also means that practitioners with less experience with the nuances of statistical inference have access to powerful tools that could lead them astray if used improperly.

\begin{figure}
    \centering
    \includegraphics[width=\linewidth]{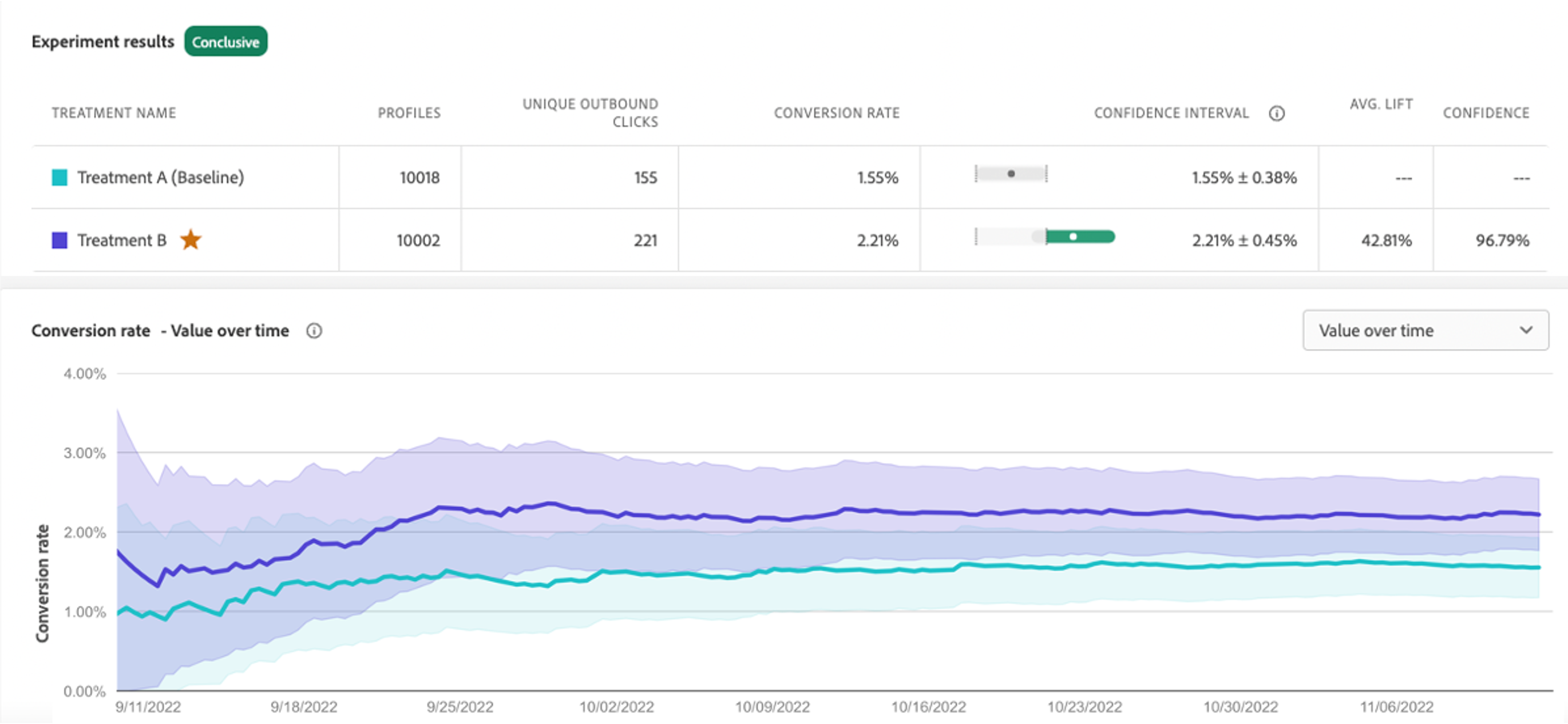}
    \caption{Experiment being analyzed with Anytime-Valid Confidence Sequences with the Experimentation Service on Adobe Experience Platform.}
    \label{fig:prod-image}
\end{figure}

This democratization of A/B tests calls for procedures that meet the needs of practitioners, while also protecting them from statistically flawed conclusions. Practitioners often want to monitor tests continually, stop tests early, or continue tests to collect more evidence. However, such ``peeking'' or ``early stopping'' is known to inflate type-I error in naive fixed horizon methodologies~\citep{wald1945sequential,robbins1970statistical, johari2015can, johari2017peeking, johari2022always, deng2016continuous, howard2021time}. In a classical fixed-horizon A/B test the user must: (1) specify a hypothesis, (2) select a sample size (often based on minimum detectable effects (MDE) and the desired type-I and type-II error), and (3) \emph{only when the pre-specified sample size is reached}, compare the $p$-value (or confidence interval) to the appropriate threshold. Performing comparisons multiple times before reaching the sample size, or collecting more data and performing additional comparisons drastically inflates the type-I error. 

As we discuss in Section \ref{sec:related}, there are multiple approaches that allow for optional stopping, ranging from group sequential methods\cite{pocock1977group, o1979multiple, gordon1983discrete} (commonly employed in clinical trials), to Bayesian methods \cite{stucchio2015bayesian, deng2016continuous}, and more recently, a burgeoning literature on ``anytime-valid'' methods \cite{johari2015can, johari2017peeking, howard2021time, waudby2020estimating} which are based on martingale techniques. Anytime-valid procedures have several desirable properties - allowing for continuous monitoring, adaptive stopping, or continuation, while controlling type-I error at all times. Recently, \citet{waudby2021doubly} introduced so-called ``Asymptotic Confidence Sequences'' (AsympCS) --- an anytime-valid analogue of Central Limit Theorem-based confidence intervals --- that are non-parametric, asymptotically narrow, and easy to compute. This paper describes the adaptation and implementation of AsympCS into an enterprise A/B testing platform. 

Our contributions include the following: 
1) we apply asymptotic confidence sequences to average treatment effect estimation in A/B tests without recorded propensity scores, 2) we provide a sample size calculator for AsympCS-based tests,
3) we compute AsympCSs for alternate test statistics, the ``lift'' (i.e.~relative improvement between experiences A and B),  
4) using simulations, we show that AsympCS has attractive properties when compared to alternative strategies commonly used in the industry, and 5) on thousands of real A/B tests conducted within the product, we elucidate the advantages of using AsympCSs over fixed horizon tests.  

\section{Related Work}
\label{sec:related}

Traditionally, peeking-induced type-I error inflation is addressed by prespecifying the peeking times, e.g. in so-called ``Group Sequential Designs'' (GSD) \cite{pocock1977group, o1979multiple, gordon1983discrete}. 
In clinical trials, where early stopping may be necessary for safety and ethical reasons, the Food and Drug Administration~\cite{fdaguidance}, recommend pre-specification of the design and strict control of peeking times. Unfortunately, such restrictions are typically infeasible in online experimentation.

The parallel literature on \emph{confidence sequences} (CS) ---pioneered by Herbert Robbins and colleagues --- focuses on methods that do not require pre-specifying peeking times at all \citep{darling1967confidence,robbins1970statistical,lai1976boundary,lai1976confidence}. 
Formally, a CS for a parameter $\theta$ is a sequence $\bar{C}_n$ of sets, such that, $\Pr(\forall n \in \mathbb{N}^+, \theta \in \bar{C}_n) > (1 - \alpha)$.
These methods were applied to A/B tests by \citet{johari2017peeking} with an emphasis on Robbins' mixture sequential probability ratio test (mSPRT). This method has excellent properties in practice but relies on a parametric model for validity, requiring strong assumptions that are unrealistic in practice. Subsequent work focused on extending CSs to richer nonparametric problems, \citep{howard2021time} such as those for bounded random variables \citep{jun2019parameter,waudby2020estimating,orabona2021tight}. See \citet{ramdas2022game} for a more detailed survey. Most of the CS literature focuses on \emph{non-asymptotic} methods which have three major disadvantages even for fixed-horizon settings: (a) they require strong assumptions, such as a parametric model or known moment generating functions
\citep{robbins1970statistical,howard2021time,waudby2020estimating,waudby2020confidence,wang2022catoni}, (b) they are typically wider than asymptotic methods based on the central limit theorem, and (c) they take different forms for different problems, whereas the central limit theorem yields a universal and closed-form (trivial-to-compute) expression. The CSs of \citet{waudby2021doubly} overcomes these issues, at the cost of satisfying anytime validity only in an asymptotic sense \citep{waudby2021doubly}, and for this reason, we adopt their AsympCS framework in the present work.

Some experimentation solutions \citep{google_bayesian, stucchio2015bayesian} opt for Bayesian Hypothesis tests (\citep[Chapter 11]{wasserman2004all}) instead of frequentist ones. For example, the Bayes Factor has also been proposed as a way of enabling continuous monitoring in A/B tests \cite{deng2016continuous}. Of course, these take an different philosophical approach to uncertainty quantification with different statistical guarantees. In addition, note that these methods rely on parametric models, and their performance is sensitive to the choice of prior. 

\section{Asymptotic Confidence Sequences}
Next we describe AsympCSs for the average treatment effect (ATE), as well as a mechanism to compute the hypothesized sample size of the test and an AsympCS for the lift. Technical details may be found in the appendices. 

\subsection{ATE With Empirical Propensities}
\label{sec:ate}
The primary estimand considered in A/B tests is the average treatment effect (ATE). \citet[Theorem 3.1]{waudby2021doubly} provides an AsympCS for the ATE in randomized experiments using doubly robust estimation. However, if propensity scores are not recorded (a common occurrence in industry), this estimator cannot be directly computed, and the A/B test must be treated as an observational study, where \citet[Theorem 3.2]{waudby2021doubly} can be employed instead. Replacing true propensity scores with their empirical estimates results in an elegant and simple-to-compute AsympCS for the ATE:
%The resulting expression for a non-regression adjusted estimate of the ATE is
\small
\begin{align}
    &\bar{C}_{n}^{\text{Asymp}} \coloneqq \Biggl( (\hat{\mu}_{1;n} - \hat{\mu}_{0;n}) \pm \beta(n, \alpha, \rho) \times \nonumber \\
    &\sqrt{\frac{n}{n-1} \left[ \frac{n}{n_0}(\hat{\sigma}_{0;n}^2 + \hat{\mu}_{0;n}^2) + \frac{n}{n_1}(\hat{\sigma}_{1;n}^2 + \hat{\mu}_{1;n}^2) - (\hat{\mu}_{1;n} - \hat{\mu}_{0;n})^2\right]} \Biggr)
    \label{eqn:asympcs}
\end{align}
\normalsize
where $\beta(n, \alpha, \rho) = \sqrt{\frac{2(n\rho^2 + 1)}{n^2 \rho^2} \log \left (\frac{\sqrt{(n \rho^2 + 1)}}{\alpha} \right )}$, and $(\hat{\mu}_{0;n}, \hat{\mu}_{1;n})$ and $(\hat{\sigma}_{0;n}, \hat{\sigma}_{1;n})$ are the running means and standard deviations in the two treatment arms (details can be found in Appendix~\appendixref{\ref{sec:appendix_ate}}). Notice that \eqref{eqn:asympcs} only depends on individual running counts, and sample averages. These calculations are simple and parallelizable, making them well-suited to large-scale analytics engines such as Adobe Analytics.\footnote{Adobe, the Adobe logo, Adobe Experience Platform, Adobe Analytics, Adobe Journey Optimizer, and Customer Journey Analytics are either registered trademarks or trademarks of Adobe in the United States and/or other countries.} 

\subsection{Hypothesized sample size}
\label{sec:samplesize}
In online experiments, one often wishes to understand the sample size required to reach significance for a given experiment. 
In the fixed-horizon setting, this is achieved via power calculations which prescribe the necessary size of an experiment given an \emph{a priori} guess of the population variance of a metric and a desired minimum detectable effect. 
In the anytime setting we can analogously compute the \textit{hypothesized sample size}, defined as
\small
\begin{align}
\inf_n \left( \Pr(\theta_{H_0} \not\in \bar{C}_n | H_1) >= 1 - \beta\right),
\label{eq:stopping}
\end{align} 
\normalsize
\textit{i.e.}, under $H_1$ (where this is a difference in means equal to the MDE), Eq.~\ref{eq:stopping} is the sample size $n^{*}$ so that the probability of rejecting $H_0$ is at least $1 - \beta$. Here, rejection occurs when the CS $\bar{C}_n$ excludes the null effect $\theta_{H_0}$ (typically zero). Estimating \eqref{eq:stopping} proceeds by replacing $\hat{\sigma}^2$ with a prior guess made by the analyst, and then noting that the earliest possible time period is entailed by the time at which the $\beta$ quantile under $H_1$ is greater than or equal to the $1 - \alpha$ quantile under $H_0$, and that the CS under the null and alternate hypotheses are monotonic with respect to $n$.
Solving \eqref{eq:stopping} can be reduced to a simple convex optimization problem (see Appendix~\appendixref{\ref{sec:appendix_samplesize}}), for which standard root finding procedures provide efficient solutions.

\subsection{Test Statistic for Lift}
While the ATE is a common target estimand, practitioners are often interested in estimating the \emph{relative} treatment effect (or ``lift''), formally defined as $(\mu_1/\mu_0 - 1)$. Unfortunately, it is not immediately obvious how one can write an estimator for lift based on sample averages, and thus we cannot directly apply the techniques of \citet{waudby2021doubly}.  

We develop an AsympCS for this test statistic as follows: (1) Begin with the logarithm of the ratio of the treatment means. (2) Use the fact that if $X$ is a continuous random variable such that $\Pr(X < \ell) = \alpha/2$, then, for a strictly monotonic increasing function $g: \mathbb{R} \rightarrow \mathbb{R}$, $\Pr(g(X) < g(\ell)) = \alpha/2$. (3) Using the one-sided AsympCS from \cite{waudby2021doubly} monotonicity of $\log(\cdot)$, generate a one-sided AsympCS for $\log \mu_i$. (4) Generate the AsympCS for the difference $(\log \mu_1 - \log \mu_0)$ via union bounds. (5) Finally, apply another monotonic transformation $\exp\{\cdot\}$ to arrive at the AsympCS for the lift: 
\small
\begin{equation}
    \bar{C}_{n}^{\text{Lift}} \coloneqq \Biggl[ \frac{\hat{\mu}_{1;n} - \hat{\sigma}_{1;n} \beta(n_1, \alpha, \rho)}{\hat{\mu}_{0;n} + \hat{\sigma}_{0;n} \beta(n_0, \alpha, \rho)} - 1, \frac{\hat{\mu}_{1;n} + \hat{\sigma}_{1;n} \beta(n_1, \alpha, \rho)}{\hat{\mu}_{0;n} - \hat{\sigma}_{0;n} \beta(n_0, \alpha, \rho)} - 1 \Biggr],
\end{equation}
\normalsize
where $\hat{\mu}_{i;n}$, $\hat{\sigma}_{i;n}$, and $\beta(n_i,\alpha,\rho)$ are defined as before. See Appendix~\appendixref{\ref{sec:lift}} for a a detailed proof.

\section{Experiments}
\label{sec:experiments}
In the following two subsections we compare the performance of \textbf{AsympCS} defined in the previous section to several baseline methodologies. The baselines are 1) \textbf{FHT-peeking}: a classic fixed-horizon $z$-test, where the $p$-value is (improperly) recomputed at regular intervals, to simulate peeking. 2) \textbf{LDM} - Lan-deMets' Group Sequential Design \cite{gordon1983discrete, demets1994interim} with the Pocock alpha spending function. 3) \textbf{mSPRT} - the mSPRT test used in \cite{johari2017peeking}. 
4) \textbf{BHT-uninformed} - a Bayesian hypothesis test where a decision-theoretic approach is used, based on the expected loss function defined in \citep{stucchio2015bayesian}. 5) \textbf{BF-uninformed} A Bayes factor approach~\cite{deng2016continuous}. For Bayesian methods, we use Beta distributions for all priors, as would typically be the case for an industrial A/B testing platform. Full details for each method are provided in Appendix~\appendixref{\ref{sec:baselines}}. 

\subsection{Type-I Error}
\label{sec:t1error}
\begin{figure}
  \centering
  \includegraphics[width = 0.85\linewidth]{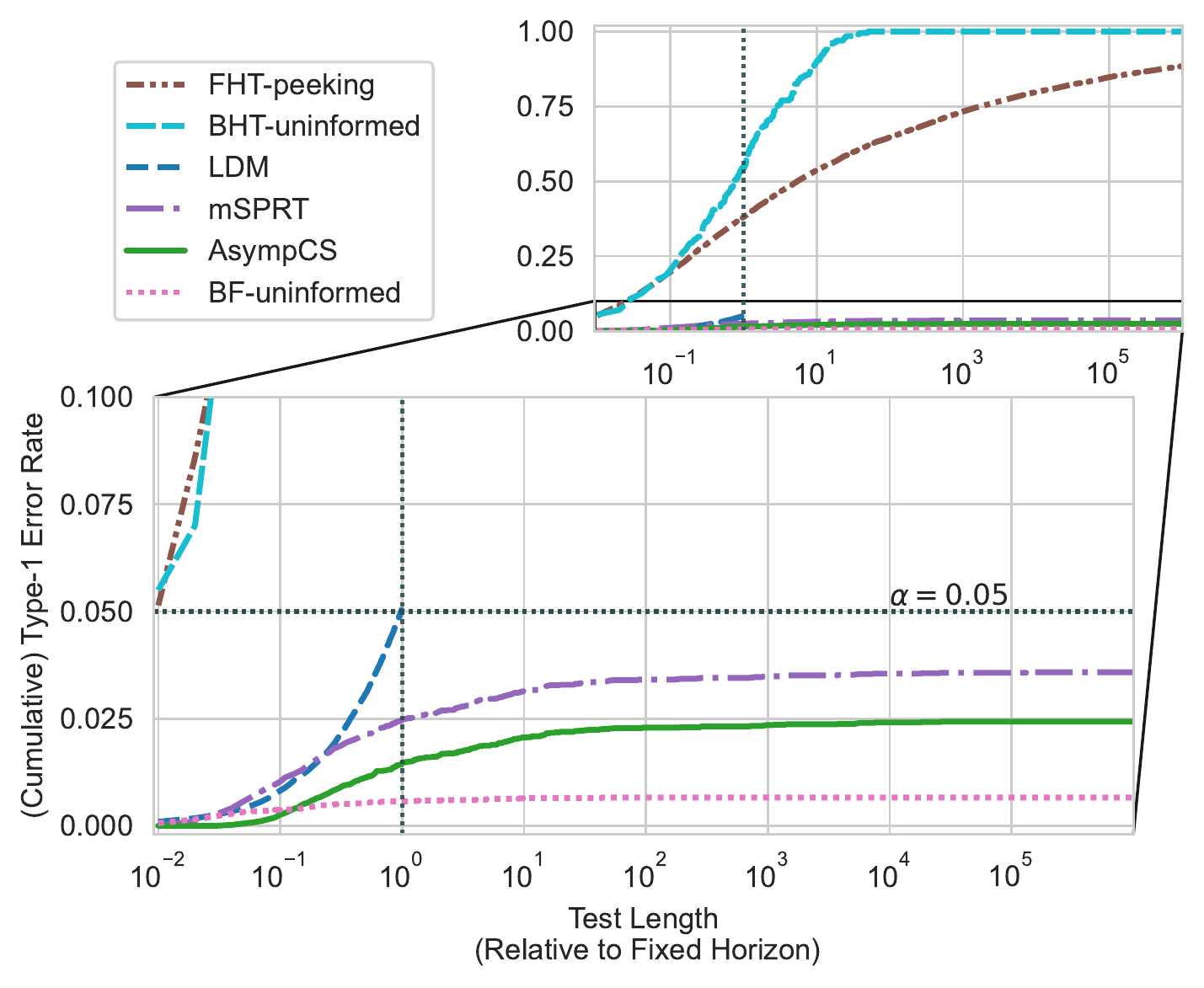}
  \vspace{-4mm}
  \caption{Cumulative type-I error over 10,000 simulated A/A experiments with binary responses, where each arm has a conversion rate of $\theta_0 = 0.1$. See Section \ref{sec:experiments} for methodology descriptions and discussion.} 
  \vspace*{-4mm}
  \label{fig:miscoverage}
\end{figure}

First, we consider the type-I error of the different approaches. We simulate 10,000 ``A/A'' tests with binary outcomes, \textit{i.e.}, tests where both arms have the same conversion rate (10\%). We continuously monitor the test, and a decision to stop is counted as a type-I error. We strongly caution that Bayesian methods are not designed to controls such a type-I error; nonetheless we believe type-I error control is valuable in general purpose industry applications. 

Figure \ref{fig:miscoverage} presents the results of these simulations. The fixed horizon run length assumes an absolute MDE of 1\%, at 80\% power and $\alpha =5\%$. As we can see, FHT-peeking does not control type-I error (asymptotically approaching $1$). BHT-uninformed does not attempt to control type-I error and indeed, it does not (see Appendix~\appendixref{\ref{sec:additional-experimental-results}} for detailed comparisons).

Conversely, the Bayes Factor (BF-uninformed) method has tight type-I error control --- in an A/A test, there is no evidence to favor a hypothesis that the test arms have different conversion rates, and Bayes Factor never exceeds a threshold of $1/\alpha$ \cite{deng2016continuous, rouder2014optional}, though this depends on choice of priors \cite{de2021optional}. The group sequential LDM method can only be run up to a fixed horizon test length, and it performs exactly as it should - with a type-I error rate of exactly 5\% at the fixed horizon. Finally, as expected, both the mSPRT and AsympCS methods control type-I error below the threshold of 5\%. However, we note that the former does not have strict type-I error guarantees outside the Gaussian setting. 

\subsection{Statistical Power and Stopping Time}
\label{sec:t2error}

\begin{figure}[t]
  \centering
  \includegraphics[width = 0.9\linewidth]{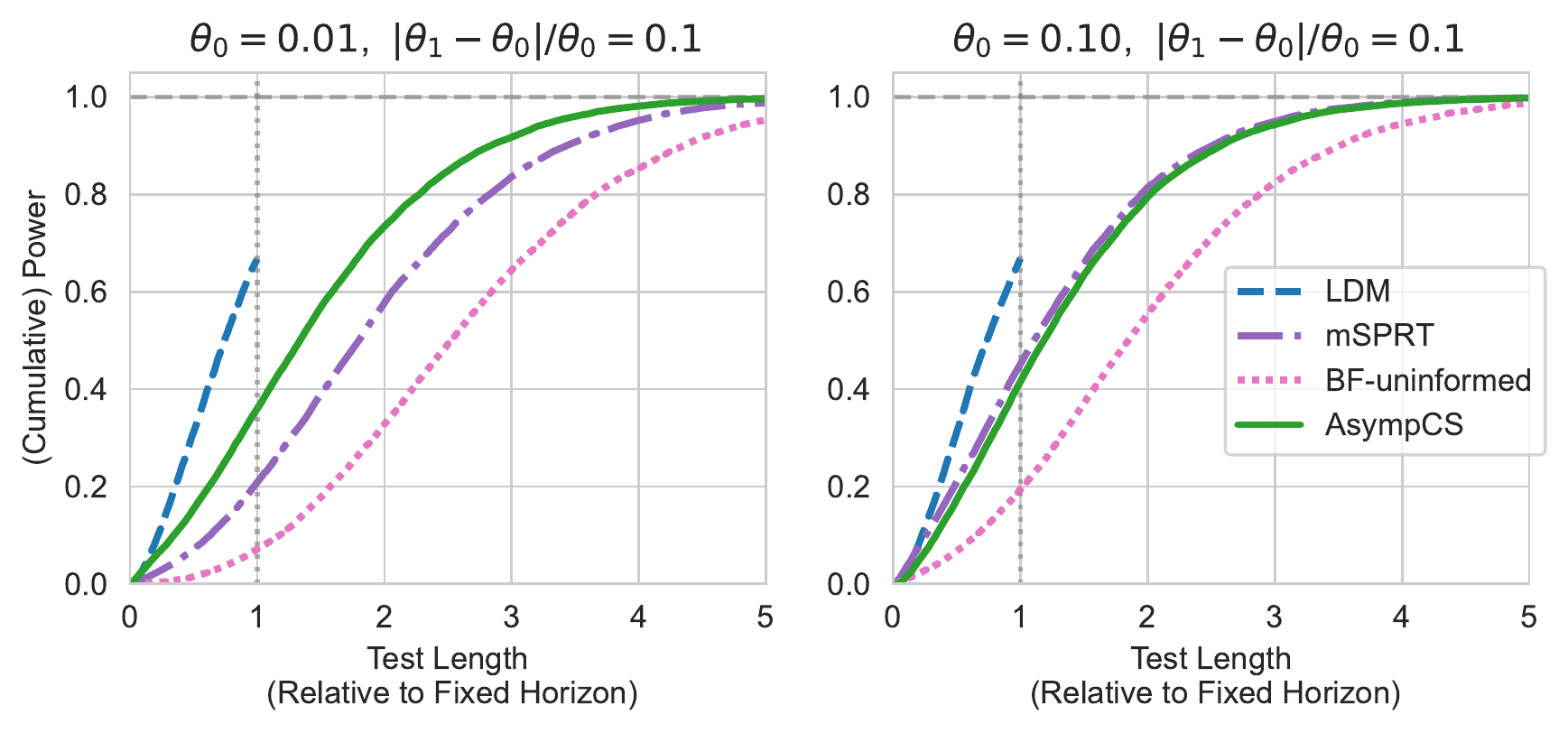}
  \vspace{-2mm}
  \caption{Statistical power for a binary response, as a function of the test length measured in multiples of the Fixed Horizon sample size at 80\% power. See Section \ref{sec:t2error} for discussion.}
    \vspace*{-5mm}
  \label{fig:power}
\end{figure}

For the methods that maintain the desired Type-I error (LDM, mSPRT, BF-uninformed, and AsympCS), we next consider their statistical power. In Figure \ref{fig:power}, we plot the statistical power of these methods against sample size (as a multiple of the FHT sample size with 80\% power). We see that LDM has the highest statistical power. Meanwhile Bayes Factor (BF-uninformed) has the lowest power, though this may improve with carefully chosen priors. 
The value of using the proposed AsympCS for the lift can be seen in Figure \ref{fig:lift_power}. While it has has lower power than the AsympCS for the ATE, we still see that the power approaches $1$ as the true lift increases. 

We also note that AsympCS and mSPRT can be tuned to have essentially the same power. When compared to the FHT sample size, we see that to achieve the same power with AsympCS, we will need two to three times the sample size. For LDM, the ``peek'' times must be recorded, and must be independent of data (challenging in a platform where many analysts may monitor data continuously) 
and moreover, we are unable to arbitrarily continue the test.

Given the rigorous error guarantees and flexible (non-parametric) assumptions of AsympCSs, combined with a comprehensive set of comparisons of both frequentist and Bayesian properties of the methods (see Appendix~\appendixref{\ref{sec:additional-experimental-results}}), we believe AsympCS is the right choice for our experimentation platform.

\subsection{Sample Size Calculator}
In Figure \ref{fig:sample_size_ratio} we show a comparison of the estimated 80th percentile stopping time (from Section \ref{sec:samplesize}) vs. an empirical stopping time (computed from 10,000 simulations). We see how the mathematically computed sample sizes closely match the empirical sample sizes, though they are conservative. We use this methodology to aid practitioners in sample size estimation in our testing platform. 

\begin{figure}[t]
  \centering
  \includegraphics[width = 0.9\linewidth]{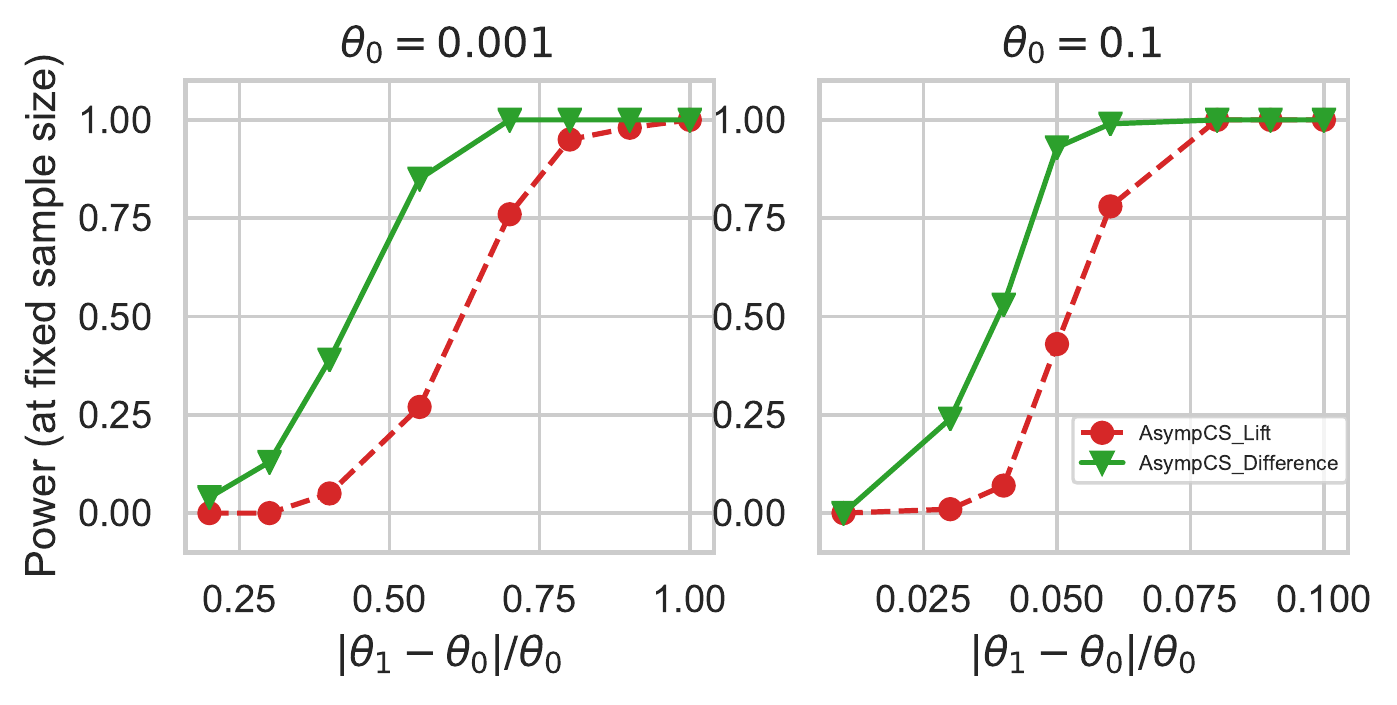}
  \vspace{-5mm}
  \caption{Power of AsympCS for the lift test statistic, which is always lower than Power for the ATE (difference in means).}
  \label{fig:lift_power}
\end{figure}
\begin{figure}
  \centering
  \includegraphics[width = 0.7\linewidth]{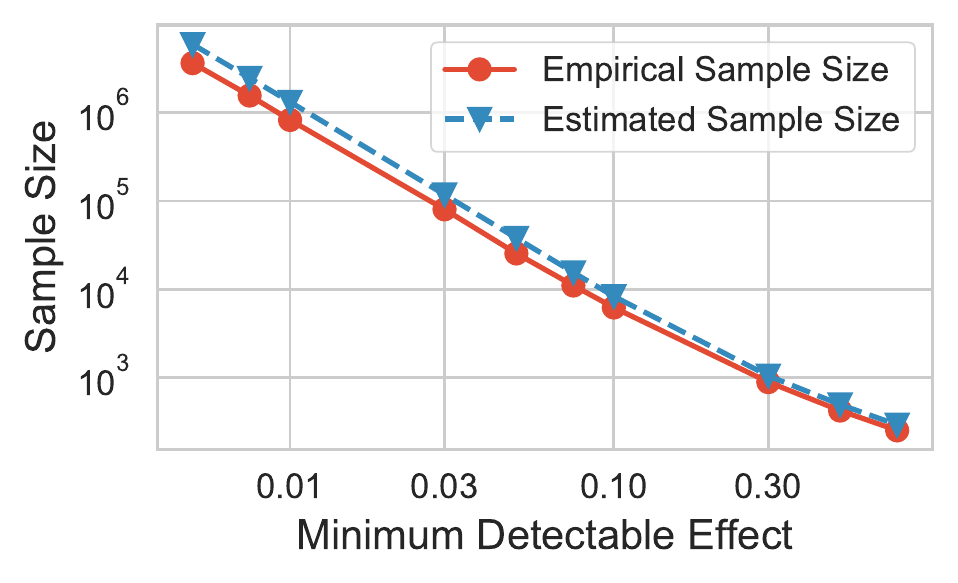}
  \vspace{-3mm}
  \caption{Comparison of estimated vs. simulated (empirical) sample size (80\% power), as a function of the MDE.}
  \label{fig:sample_size_ratio}
    \vspace*{-3mm}
\end{figure}

\subsection{Real-World Experiments}
We consider 2,089 real A/B tests conducted on our experimentation service. Similar to what we have observed in simulations, using FHT with ``peeking'' leads to erroneous conclusions in real experiments. As seen in Table~\ref{tab:fixed_asympcs}, $57\%$ of experiments resulted in non-significant tests using both FHT and AsympCS, and $28\%$ of experiments resulted in significant tests using both FHT and AsympCS (and hence there being agreement on the conclusion for $85\%$ of the experiments). Notice that $15\%$ of the time, a FHT with continuous monitoring resulted in statistical significance while the AsympCS did not. 

\begin{table}[t]
\small
\caption{Comparison of AsympCS and FHT on real A/B Tests}
\label{tab:fixed_asympcs}
\resizebox{\columnwidth}{!}{%
\begin{tabular}{l|cc|c}
& AsympCS Significant & AsympCS Not Significant & Total \\
\hline
FHT Significant     & 28\% (593)          & 15\% (308)              & 901   \\
FHT Not Significant & 0.1\% (3)           & 57\% (1185)             & 1188  \\
\hline
Total               & 596                 & 1493                    & 2089 
\end{tabular}%
}
\end{table}

\begin{figure}[t]
  \centering
  \includegraphics[width = 0.85\linewidth]{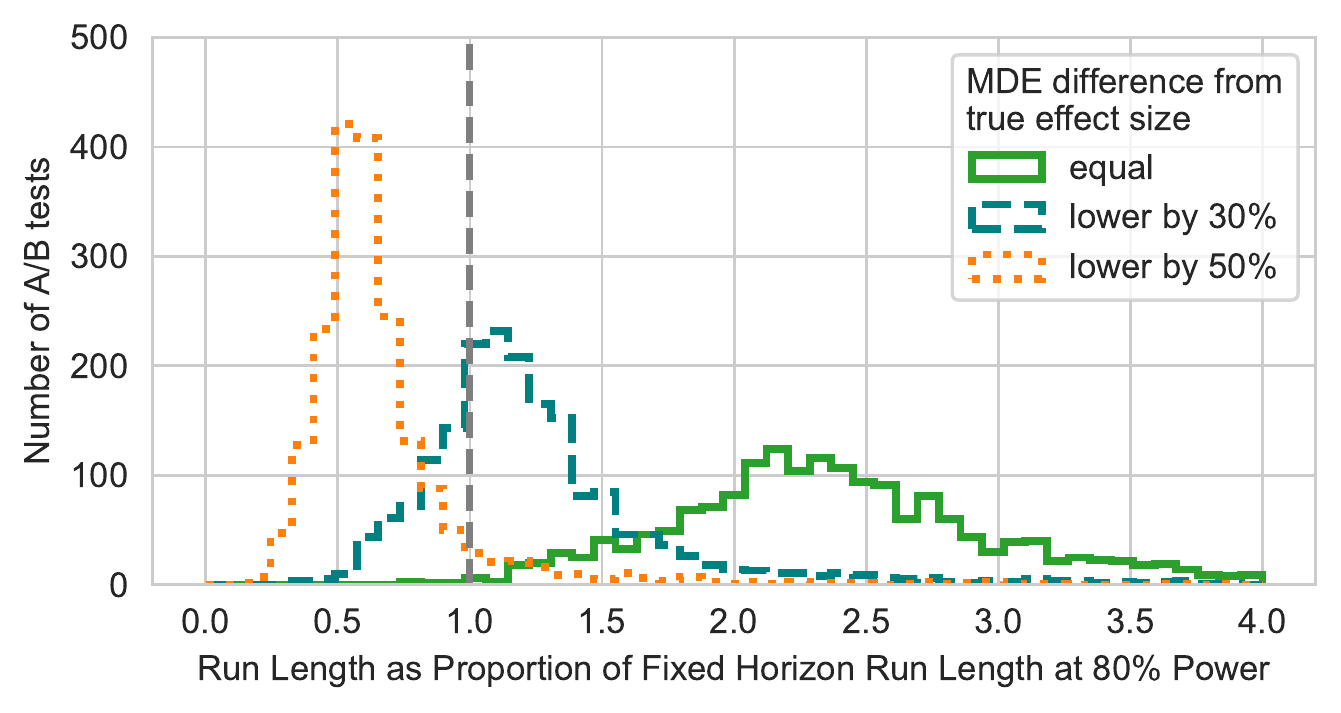}
  \vspace{-3mm}
  \caption{Sample sizes of AsympCS methods that would be observed, given real A/B test effect sizes.}
  \label{fig:run_lengths}
      \vspace*{-5mm}
\end{figure}

A commonly cited downside of anytime-valid methods is that they take longer to reach a conclusion compared to FHTs with optimally tuned sample sizes --- for example, if the Minimum Detectable Effect (MDE) is known \emph{a priori}, then a FHT with a sample size calculation will stop earlier than an AsympCS-based test. Note however that if the ATE and variance are not known (and they are almost always unknown), then it is not reasonable to expect to know the MDE. In such cases, an AsympCS can stop substantially earlier and with higher power. For example, if the MDE is under-estimated, then a FHT could be under-powered, resulting in false negatives and requiring experiments to be restarted (wasting the collected data), while if the MDE is over-estimated, a FHT will waste time and resources (e.g. requiring 10k samples when 1k would have sufficed). AsympCSs sidestep the need for MDE calculation altogether, allowing A/B tests to run only as long as they need to in order to achieve significance --- no more, no less. This phenomenon is illustrated in Figure \ref{fig:run_lengths} for a real world distribution of ATEs. We took the final observed ATE of 2,089 real world A/B tests, and simulated 1000 tests (with this ATE) for each real one, and compared the stopping times at 80\% power of each method. We see that if the MDE is guessed correctly, then an AsympCS does in fact take longer to stop, But if the true MDE is 50\% smaller than anticipated, an AsympCS stops roughly twice as fast as an FHT. 
\vspace*{-2mm}

\subsection{Implementation on Adobe's Test Platform}
The AsympCS-based analysis of online A/B tests is now available in the ``Adobe Journey Optimizer'' \citep{Getstart17:online, Statisti60:online}. In addition, the same technology is available in Customer Journey Analytics \citep{cja}, a flexible data analysis ecosystem where customers can connect to their own experimental data and analyze it using our approach. Both products employ the same underlying time-series database that allows for ad-hoc creation of metrics, dimensions and segments, and can query web-scale data in seconds. The fact that AsympCSs only require sample means and variances of metrics makes implementation straightforward and computation fast.  

\section{Conclusion}
We have proposed and implemented an anytime-valid strategy for online A/B tests based on asymptotic confidence sequences, permitting continuous monitoring of experimental results. Our approach enjoys rigorous anytime guarantees, can be easily implemented on web-scale analytics engines, can estimate alternate metrics like lift, and provides a sample size calculator for the design of A/B tests. We illustrated the advantages of our approach using simulations, and on several real A/B tests conducted on our platform. In the future, we plan to explore anytime-valid approaches to handling multi-armed bandits and time-to-event data. 

\bibliographystyle{ACM-Reference-Format}
\bibliography{bib}


\begin{thebibliography}{31}

%%% the Web address.

\ifx \showCODEN    \undefined \def \showCODEN     #1{\unskip}     \fi
\ifx \showDOI      \undefined \def \showDOI       #1{#1}\fi
\ifx \showISBNx    \undefined \def \showISBNx     #1{\unskip}     \fi
\ifx \showISBNxiii \undefined \def \showISBNxiii  #1{\unskip}     \fi
\ifx \showISSN     \undefined \def \showISSN      #1{\unskip}     \fi
\ifx \showLCCN     \undefined \def \showLCCN      #1{\unskip}     \fi
\ifx \shownote     \undefined \def \shownote      #1{#1}          \fi
\ifx \showarticletitle \undefined \def \showarticletitle #1{#1}   \fi
\ifx \showURL      \undefined \def \showURL       {\relax}        \fi
% invisible to TeX
\providecommand\bibfield[2]{#2}
\providecommand\bibinfo[2]{#2}
\providecommand\natexlab[1]{#1}
\providecommand\showeprint[2][]{arXiv:#2}

\bibitem[Adobe(2022a)]%
        {cja}
\bibfield{author}{\bibinfo{person}{Adobe}.} \bibinfo{year}{2022}\natexlab{a}.
\newblock \bibinfo{title}{Experimentation panel, Adobe Customer Journey
  Analytics}.
\newblock
  \bibinfo{howpublished}{\url{https://experienceleague.adobe.com/docs/analytics-platform/using/cja-workspace/panels/experimentation.html?lang=en}}.
\newblock
\newblock
\shownote{(Accessed on 11/07/2022)}.


\bibitem[Adobe(2022b)]%
        {Getstart17:online}
\bibfield{author}{\bibinfo{person}{Adobe}.} \bibinfo{year}{2022}\natexlab{b}.
\newblock \bibinfo{title}{Get started with content experiment | Adobe Journey
  Optimizer}.
\newblock
  \bibinfo{howpublished}{\url{https://experienceleague.adobe.com/docs/journey-optimizer/using/campaigns/content-experiment/get-started-experiment.html}}.
\newblock
\newblock
\shownote{(Accessed on 11/07/2022)}.


\bibitem[Adobe(2022c)]%
        {Statisti60:online}
\bibfield{author}{\bibinfo{person}{Adobe}.} \bibinfo{year}{2022}\natexlab{c}.
\newblock \bibinfo{title}{Statistical Calculations used by Adobe Journey
  Optimizer Experimentation | Adobe Journey Optimizer}.
\newblock
  \bibinfo{howpublished}{\url{https://experienceleague.adobe.com/docs/journey-optimizer/using/campaigns/content-experiment/experiment-calculations.html?lang=en}}.
\newblock
\newblock
\shownote{(Accessed on 11/07/2022)}.


\bibitem[Darling and Robbins(1967)]%
        {darling1967confidence}
\bibfield{author}{\bibinfo{person}{DA Darling} {and} \bibinfo{person}{Herbert
  Robbins}.} \bibinfo{year}{1967}\natexlab{}.
\newblock \showarticletitle{Confidence sequences for mean, variance, and
  median}.
\newblock \bibinfo{journal}{\emph{Proceedings of the National Academy of
  Sciences of the United States of America}} \bibinfo{volume}{58},
  \bibinfo{number}{1} (\bibinfo{year}{1967}), \bibinfo{pages}{66}.
\newblock


\bibitem[De~Heide and Gr{\"u}nwald(2021)]%
        {de2021optional}
\bibfield{author}{\bibinfo{person}{Rianne De~Heide} {and}
  \bibinfo{person}{Peter~D Gr{\"u}nwald}.} \bibinfo{year}{2021}\natexlab{}.
\newblock \showarticletitle{Why optional stopping can be a problem for
  Bayesians}.
\newblock \bibinfo{journal}{\emph{Psychonomic Bulletin \& Review}}
  \bibinfo{volume}{28}, \bibinfo{number}{3} (\bibinfo{year}{2021}),
  \bibinfo{pages}{795--812}.
\newblock


\bibitem[Demets and Lan(1994)]%
        {demets1994interim}
\bibfield{author}{\bibinfo{person}{David~L Demets} {and}
  \bibinfo{person}{KK~Gordon Lan}.} \bibinfo{year}{1994}\natexlab{}.
\newblock \showarticletitle{Interim analysis: the alpha spending function
  approach}.
\newblock \bibinfo{journal}{\emph{Statistics in medicine}}
  \bibinfo{volume}{13}, \bibinfo{number}{13-14} (\bibinfo{year}{1994}),
  \bibinfo{pages}{1341--1352}.
\newblock


\bibitem[Deng et~al\mbox{.}(2016)]%
        {deng2016continuous}
\bibfield{author}{\bibinfo{person}{Alex Deng}, \bibinfo{person}{Jiannan Lu},
  {and} \bibinfo{person}{Shouyuan Chen}.} \bibinfo{year}{2016}\natexlab{}.
\newblock \showarticletitle{Continuous monitoring of A/B tests without pain:
  Optional stopping in Bayesian testing}. In \bibinfo{booktitle}{\emph{2016
  IEEE international conference on data science and advanced analytics
  (DSAA)}}. IEEE, \bibinfo{pages}{243--252}.
\newblock


\bibitem[Food and Administration(2018)]%
        {fdaguidance}
\bibfield{author}{\bibinfo{person}{Food} {and} \bibinfo{person}{Drug
  Administration}.} \bibinfo{year}{2018}\natexlab{}.
\newblock \bibinfo{booktitle}{\emph{Adaptive Design Clinical Trials for Drugs
  and Biologics Guidance for Industry}}.
\newblock \bibinfo{type}{Guidance Document} FDA-2018-D-3124.
  \bibinfo{institution}{U.S. Department of Health and Human Services}.
\newblock


\bibitem[Gordon~Lan and DeMets(1983)]%
        {gordon1983discrete}
\bibfield{author}{\bibinfo{person}{KK Gordon~Lan} {and}
  \bibinfo{person}{David~L DeMets}.} \bibinfo{year}{1983}\natexlab{}.
\newblock \showarticletitle{Discrete sequential boundaries for clinical
  trials}.
\newblock \bibinfo{journal}{\emph{Biometrika}} \bibinfo{volume}{70},
  \bibinfo{number}{3} (\bibinfo{year}{1983}), \bibinfo{pages}{659--663}.
\newblock


\bibitem[Help(2022)]%
        {google_bayesian}
\bibfield{author}{\bibinfo{person}{Technical Help}.}
  \bibinfo{year}{2022}\natexlab{}.
\newblock \bibinfo{title}{General methodology - optimize resource hub}.
\newblock
\newblock
\urldef\tempurl%
\url{https://support.google.com/optimize/answer/7405543?hl=en#zippy=\%2Cin-this-article}
\showURL{%
\tempurl}


\bibitem[Howard et~al\mbox{.}(2021)]%
        {howard2021time}
\bibfield{author}{\bibinfo{person}{Steven~R Howard}, \bibinfo{person}{Aaditya
  Ramdas}, \bibinfo{person}{Jon McAuliffe}, {and} \bibinfo{person}{Jasjeet
  Sekhon}.} \bibinfo{year}{2021}\natexlab{}.
\newblock \showarticletitle{Time-uniform, nonparametric, nonasymptotic
  confidence sequences}.
\newblock \bibinfo{journal}{\emph{The Annals of Statistics}}
  \bibinfo{volume}{49}, \bibinfo{number}{2} (\bibinfo{year}{2021}),
  \bibinfo{pages}{1055--1080}.
\newblock


\bibitem[Johari(2015)]%
        {johari2015can}
\bibfield{author}{\bibinfo{person}{Ramesh Johari}.}
  \bibinfo{year}{2015}\natexlab{}.
\newblock \showarticletitle{Can I Take a Peek? Continuous Monitoring of Online
  A/B Tests}. In \bibinfo{booktitle}{\emph{Proceedings of the 24th
  International Conference on World Wide Web}}. \bibinfo{pages}{915--915}.
\newblock


\bibitem[Johari et~al\mbox{.}(2017)]%
        {johari2017peeking}
\bibfield{author}{\bibinfo{person}{Ramesh Johari}, \bibinfo{person}{Pete
  Koomen}, \bibinfo{person}{Leonid Pekelis}, {and} \bibinfo{person}{David
  Walsh}.} \bibinfo{year}{2017}\natexlab{}.
\newblock \showarticletitle{Peeking at a/b tests: Why it matters, and what to
  do about it}. In \bibinfo{booktitle}{\emph{Proceedings of the 23rd ACM SIGKDD
  International Conference on Knowledge Discovery and Data Mining}}.
  \bibinfo{pages}{1517--1525}.
\newblock


\bibitem[Johari et~al\mbox{.}(2022)]%
        {johari2022always}
\bibfield{author}{\bibinfo{person}{Ramesh Johari}, \bibinfo{person}{Pete
  Koomen}, \bibinfo{person}{Leonid Pekelis}, {and} \bibinfo{person}{David
  Walsh}.} \bibinfo{year}{2022}\natexlab{}.
\newblock \showarticletitle{Always valid inference: Continuous monitoring of
  a/b tests}.
\newblock \bibinfo{journal}{\emph{Operations Research}} \bibinfo{volume}{70},
  \bibinfo{number}{3} (\bibinfo{year}{2022}), \bibinfo{pages}{1806--1821}.
\newblock


\bibitem[Jun and Orabona(2019)]%
        {jun2019parameter}
\bibfield{author}{\bibinfo{person}{Kwang-Sung Jun} {and}
  \bibinfo{person}{Francesco Orabona}.} \bibinfo{year}{2019}\natexlab{}.
\newblock \showarticletitle{Parameter-free online convex optimization with
  sub-exponential noise}. In \bibinfo{booktitle}{\emph{Conference on Learning
  Theory}}. PMLR, \bibinfo{pages}{1802--1823}.
\newblock


\bibitem[Kohavi et~al\mbox{.}(2020)]%
        {kohavi2020trustworthy}
\bibfield{author}{\bibinfo{person}{Ron Kohavi}, \bibinfo{person}{Diane Tang},
  {and} \bibinfo{person}{Ya Xu}.} \bibinfo{year}{2020}\natexlab{}.
\newblock \bibinfo{booktitle}{\emph{Trustworthy online controlled experiments:
  A practical guide to a/b testing}}.
\newblock \bibinfo{publisher}{Cambridge University Press}.
\newblock


\bibitem[Lai(1976a)]%
        {lai1976boundary}
\bibfield{author}{\bibinfo{person}{Tze~Leung Lai}.}
  \bibinfo{year}{1976}\natexlab{a}.
\newblock \showarticletitle{Boundary crossing probabilities for sample sums and
  confidence sequences}.
\newblock \bibinfo{journal}{\emph{The Annals of Probability}}
  \bibinfo{volume}{4}, \bibinfo{number}{2} (\bibinfo{year}{1976}),
  \bibinfo{pages}{299--312}.
\newblock


\bibitem[Lai(1976b)]%
        {lai1976confidence}
\bibfield{author}{\bibinfo{person}{Tze~Leung Lai}.}
  \bibinfo{year}{1976}\natexlab{b}.
\newblock \showarticletitle{On confidence sequences}.
\newblock \bibinfo{journal}{\emph{The Annals of Statistics}}
  (\bibinfo{year}{1976}), \bibinfo{pages}{265--280}.
\newblock


\bibitem[O'Brien and Fleming(1979)]%
        {o1979multiple}
\bibfield{author}{\bibinfo{person}{Peter~C O'Brien} {and}
  \bibinfo{person}{Thomas~R Fleming}.} \bibinfo{year}{1979}\natexlab{}.
\newblock \showarticletitle{A multiple testing procedure for clinical trials}.
\newblock \bibinfo{journal}{\emph{Biometrics}} (\bibinfo{year}{1979}),
  \bibinfo{pages}{549--556}.
\newblock


\bibitem[Orabona and Jun(2021)]%
        {orabona2021tight}
\bibfield{author}{\bibinfo{person}{Francesco Orabona} {and}
  \bibinfo{person}{Kwang-Sung Jun}.} \bibinfo{year}{2021}\natexlab{}.
\newblock \showarticletitle{Tight concentrations and confidence sequences from
  the regret of universal portfolio}.
\newblock \bibinfo{journal}{\emph{arXiv preprint arXiv:2110.14099}}
  (\bibinfo{year}{2021}).
\newblock


\bibitem[Pocock(1977)]%
        {pocock1977group}
\bibfield{author}{\bibinfo{person}{Stuart~J Pocock}.}
  \bibinfo{year}{1977}\natexlab{}.
\newblock \showarticletitle{Group sequential methods in the design and analysis
  of clinical trials}.
\newblock \bibinfo{journal}{\emph{Biometrika}} \bibinfo{volume}{64},
  \bibinfo{number}{2} (\bibinfo{year}{1977}), \bibinfo{pages}{191--199}.
\newblock


\bibitem[Ramdas et~al\mbox{.}(2022)]%
        {ramdas2022game}
\bibfield{author}{\bibinfo{person}{Aaditya Ramdas}, \bibinfo{person}{Peter
  Gr{\"u}nwald}, \bibinfo{person}{Vladimir Vovk}, {and} \bibinfo{person}{Glenn
  Shafer}.} \bibinfo{year}{2022}\natexlab{}.
\newblock \showarticletitle{Game-theoretic statistics and safe anytime-valid
  inference}.
\newblock \bibinfo{journal}{\emph{arXiv preprint arXiv:2210.01948}}
  (\bibinfo{year}{2022}).
\newblock


\bibitem[Robbins(1970)]%
        {robbins1970statistical}
\bibfield{author}{\bibinfo{person}{Herbert Robbins}.}
  \bibinfo{year}{1970}\natexlab{}.
\newblock \showarticletitle{Statistical methods related to the law of the
  iterated logarithm}.
\newblock \bibinfo{journal}{\emph{The Annals of Mathematical Statistics}}
  \bibinfo{volume}{41}, \bibinfo{number}{5} (\bibinfo{year}{1970}),
  \bibinfo{pages}{1397--1409}.
\newblock


\bibitem[Rouder(2014)]%
        {rouder2014optional}
\bibfield{author}{\bibinfo{person}{Jeffrey~N Rouder}.}
  \bibinfo{year}{2014}\natexlab{}.
\newblock \showarticletitle{Optional stopping: No problem for Bayesians}.
\newblock \bibinfo{journal}{\emph{Psychonomic bulletin \& review}}
  \bibinfo{volume}{21}, \bibinfo{number}{2} (\bibinfo{year}{2014}),
  \bibinfo{pages}{301--308}.
\newblock


\bibitem[Stucchio(2015)]%
        {stucchio2015bayesian}
\bibfield{author}{\bibinfo{person}{Chris Stucchio}.}
  \bibinfo{year}{2015}\natexlab{}.
\newblock \showarticletitle{Bayesian A/B testing at VWO}.
\newblock \bibinfo{journal}{\emph{Whitepaper, Visual Website Optimizer}}
  (\bibinfo{year}{2015}).
\newblock


\bibitem[Wald(1945)]%
        {wald1945sequential}
\bibfield{author}{\bibinfo{person}{Abraham Wald}.}
  \bibinfo{year}{1945}\natexlab{}.
\newblock \showarticletitle{Sequential tests of statistical hypotheses}.
\newblock \bibinfo{journal}{\emph{The Annals of mathematical statistics}}
  \bibinfo{volume}{16}, \bibinfo{number}{2} (\bibinfo{year}{1945}),
  \bibinfo{pages}{117--186}.
\newblock


\bibitem[Wang and Ramdas(2022)]%
        {wang2022catoni}
\bibfield{author}{\bibinfo{person}{Hongjian Wang} {and}
  \bibinfo{person}{Aaditya Ramdas}.} \bibinfo{year}{2022}\natexlab{}.
\newblock \showarticletitle{Catoni-style confidence sequences for heavy-tailed
  mean estimation}.
\newblock \bibinfo{journal}{\emph{arXiv preprint arXiv:2202.01250}}
  (\bibinfo{year}{2022}).
\newblock


\bibitem[Wasserman(2004)]%
        {wasserman2004all}
\bibfield{author}{\bibinfo{person}{Larry Wasserman}.}
  \bibinfo{year}{2004}\natexlab{}.
\newblock \bibinfo{booktitle}{\emph{All of statistics: a concise course in
  statistical inference}}. Vol.~\bibinfo{volume}{26}.
\newblock \bibinfo{publisher}{Springer}.
\newblock


\bibitem[Waudby-Smith et~al\mbox{.}(2021)]%
        {waudby2021doubly}
\bibfield{author}{\bibinfo{person}{Ian Waudby-Smith}, \bibinfo{person}{David
  Arbour}, \bibinfo{person}{Ritwik Sinha}, \bibinfo{person}{Edward~H Kennedy},
  {and} \bibinfo{person}{Aaditya Ramdas}.} \bibinfo{year}{2021}\natexlab{}.
\newblock \showarticletitle{Time-uniform central limit theory, asymptotic
  confidence sequences, and anytime-valid causal inference}.
\newblock \bibinfo{journal}{\emph{arXiv preprint arXiv:2103.06476}}
  (\bibinfo{year}{2021}).
\newblock


\bibitem[Waudby-Smith and Ramdas(2020)]%
        {waudby2020confidence}
\bibfield{author}{\bibinfo{person}{Ian Waudby-Smith} {and}
  \bibinfo{person}{Aaditya Ramdas}.} \bibinfo{year}{2020}\natexlab{}.
\newblock \showarticletitle{Confidence sequences for sampling without
  replacement}.
\newblock \bibinfo{journal}{\emph{Advances in Neural Information Processing
  Systems}}  \bibinfo{volume}{33} (\bibinfo{year}{2020}),
  \bibinfo{pages}{20204--20214}.
\newblock


\bibitem[Waudby-Smith and Ramdas(2023)]%
        {waudby2020estimating}
\bibfield{author}{\bibinfo{person}{Ian Waudby-Smith} {and}
  \bibinfo{person}{Aaditya Ramdas}.} \bibinfo{year}{2023}\natexlab{}.
\newblock \showarticletitle{Estimating means of bounded random variables by
  betting}.
\newblock \bibinfo{journal}{\emph{Journal of the Royal Statistical Society,
  Series B (to appear with discussion)}} (\bibinfo{year}{2023}).
\newblock


\end{thebibliography}

\onecolumn
\appendix

\section{Appendix}

\subsection{Average Treatment Effect With Empirical Propensities}
\label{sec:appendix_ate}
In this section, we show how the inverse propensity weighted estimator can be refined, by replacing propensities with ``empirical'' propensities. 

For each user $i$ in the experiment with covariates $X_i$, and who has assignment $A_{i} = {0, 1}$ to arms $0$ or $1$ respectively in the experiment, with propensity $\pi(X_i)$, \citet{waudby2021doubly} considered the inverse propensity weighted (IPW) estimator, defined as:
\begin{align}
\hat{\psi}_n &= \frac{1}{n} \sum^{n}_{i=1} \hat{f}_{n^{\prime}}(Z_i) \nonumber
=  \frac{1}{n}\sum^{n}_{i=1} \left(\frac{A_i}{\pi_{n^{\prime}}(X_i)} - \frac{1 - A_i}{1 - \pi_{n^{\prime}}(X_i)}\right)Y_i,
\end{align}
where $n^{\prime}$ denotes that the estimator is based on the first $n^{\prime} < n$ observations (recall, the scenario there was sample splitting being used to estimate a regression adjusted inverse propensity weighted estimator. For our purposes, we omit regression adjustments, and will set $n^{\prime} = n$). In~\cite{waudby2021doubly}, they showed that, the inverse propensity weighted estimator has a confidence sequence of the form:
\begin{align}
\bar{C}_{n} &= \hat{\psi}_{n} \pm \sqrt{\widehat{\text{var}}_{n}(\hat{f}_{n})}\cdot \sqrt{\frac{2(n\rho^2 + 1)}{n^2\rho^2}\log\left(\frac{\sqrt{n\rho^2 + 1}}{\alpha}\right)}.
\end{align}
For our purposes, it would be useful to derive an explicit expression for the variance of the inverse propensity weighted estimator, in terms of the more commonly used individual sample variances.

\subsubsection{Refinement of IPW estimator for fixed probabilities}

Suppose after $n$ observations, we have $n_1$ observations in arm 1, each with value $Y^{1}_{i_1}$, and $n_0$ observations in arm 0, with values $Y^{0}_{i_0}$. In this section,
we will assume fixed propensities which are equal to the empirical probability of being assigned to a given treatment arm, \textit{i.e.}, we have $\pi_{n}(X_i) = \pi = n_1/ (n_0 + n_1) = n_1 / n$. Note that this means we are setting $T^{\prime} = T$, \textit{i.e.}, the estimand is based on the observed empirical probability of treatment assignment at time $T$. With this, we get:
\begin{align}
\hat{\psi}_n &= \frac{1}{n}\sum^{n}_{i=1}\hat{f}(Z_i)\nonumber =\frac{1}{n}\sum^{n}_{i=1}\left(\frac{n}{n_1}A_i  - \frac{n}{n_0}(1-A_i)\right)Y_{i} \nonumber= \frac{1}{n_1}\sum^{n_1}_{i_{1} = 1} Y^{1}_{i_1} - \frac{1}{n_0} \sum^{n_0}_{i_0 = 1} Y^{0}_{i_0} \nonumber = \hat{\mu}_{1;n} - \hat{\mu}_{0;n}.
\end{align}

So we have shown that the IPW estimator reduces to a regular difference in means if we choose propensities to be equal to their empirical weights. 

\subsubsection{Variance calculation for IPW Estimator}

With this refinement of the IPW Estimator in scenarios where we use empirical probabilities as propensities, we now obtain an expression for the variance of the IPW Estimator in terms of individual sample sizes, means and variances.  First, we need to calculate $\widehat{\text{var}}_{n}(\hat{f})$, which as for any regular unbiased variance estimate, has the following form:
\begin{align}
\widehat{\text{var}}_{n}(\hat{f}_{n}) &= \frac{1}{n - 1}\sum^{n}_{i=1}\left(\hat{f}(Z_i) - \overline{\hat{f}(Z_i)}\right)^2
= \frac{n}{n-1}\left[ \left(\frac{1}{n}\sum^{n}_{i=1}\hat{f}^2(Z_i)\right) - \left(\hat{\mu}^{1}_{n_1} - \hat{\mu}^{0}_{n_0}\right)^2\right]
\end{align}

As $A_i = 0 \text{ or } 1$, we can simplify the sum of squares term as follows:
% we have the following identities:

% \begin{align}
% A^2_i &= A_i\\
% (1-A_i)^2 &= (1-A_i)\\
% A_i(1-A_i) &= 0
% \end{align}

% With these identities we can simplify the sum of squares term as follows:

\begin{align}
\frac{1}{n}\sum^{n}_{i=1}\hat{f}^2(Z_i) &= \frac{1}{n}\sum^{n}_{i=1}\left(\frac{n}{n_1}A_i  - \frac{n}{n_0}(1-A_i)\right)^2Y^2_{i} = \frac{n}{n^2_1} \sum^{n_1}_{i_1=1}Y^{2}_{i_1}  + \frac{n}{n^2_0} \sum^{n_0}_{i_0=1}Y^{2}_{i_0}\\
\end{align}

Then, defining the (biased) sample variance $\hat{\sigma}^2_0$ and $\hat{\sigma}^2_1$  for each arm:

\begin{align}
\hat{\sigma}^2_{1;n} = \left(\frac{1}{n_1}\sum^{n_1}_{i_1=1}Y^{2}_{i_1}\right) - \left(\hat{\mu}_{1;n}\right)^2
\end{align}
and similarly,
\begin{align}
\hat{\sigma}^2_{0;n} &= \left(\frac{1}{n_0}\sum^{n_0}_{i_0=1}Y^{2}_{i_0}\right) - \left(\hat{\mu}_{0;n}\right)^2.
\end{align}
Simplifying our notation, and writing $\hat{\mu}^{j}_{n_j} = \hat{\mu}_j$,
we obtain the variance of the inverse propensity weighted estimator, expressed in terms of individual sample means and variances:
\begin{align}
\widehat{\text{var}}_{n}(\hat{f}_{n}) &=  \frac{n}{n-1}\left[ \frac{n}{n_0}\left(\hat{\sigma}^2_{0;n} + \hat{\mu}^2_{0;n}\right) + \frac{n}{n_1}\left(\hat{\sigma}^2_{1;n} + \hat{\mu}^2_{1;n}\right) - \left(\hat{\mu}_{1;n} - \hat{\mu}_{0;n}\right)^2\right]
\end{align}

 We therefore have the following final expression for the confidence sequence for the average treatment effect after $n$ users, with $n_0$ and $n_1$ in arms 0 and 1 of the experiment:
\begin{align}
\bar{C}_{n}^{\text{AsympCS}} := \left(\hat{\mu}_{1;n} - \hat{\mu}_{0;n}\right) \pm \sqrt{\frac{n}{n-1}\left[ \frac{n}{n_0}\left(\hat{\sigma}^2_{0;n} + \hat{\mu}^2_{0;n}\right) + \frac{n}{n_1}\left(\hat{\sigma}^2_{1;n} + \hat{\mu}^2_{1;n}\right) - \left(\hat{\mu}_{1;n} - \hat{\mu}_{0;n}\right)^2\right]}\cdot \sqrt{\frac{2(n\rho^2 + 1)}{n^2\rho^2}\log\left(\frac{\sqrt{n\rho^2 + 1}}{\alpha}\right)}
\end{align}

%%%%%%%%%%%%%%%%%%%%%%%%%%%%%%%%%%%%%%%%%%%%%%%%%%%%%%%%%%%%%%%%%%%%%%%%%%%%%%%%%%%%%%%%%%%%%%%%%%%
%%%%%%%%%%%%%%%%%%%%%%%%%%%%%%%%%%%%%%%%%%%%%%%%%%%%%%%%%%%%%%%%%%%%%%%%%%%%%%%%%%%%%%%%%%%%%%%%%%%
%%%%%%%%%%%%%%%%%%%%%%%%%%%%%%%%%%%%%%%%%%%%%%%%%%%%%%%%%%%%%%%%%%%%%%%%%%%%%%%%%%%%%%%%%%%%%%%%%%%

\subsection{Sample Size Calculations}
\label{sec:appendix_samplesize}

In this section, we briefly review the notion of power in experiments and discuss its analogue in the AsympCS framework, with a focus on sample size estimation. In particular, we will derive an estimate for the expected stopping size of an anytime valid test which is an analogue of the required sample size for a fixed horizon test. Sample size calculation is a critical tool in the design, deployment, and analysis of experiments and allows practitioners to reason over the sizing of their experiments as well as the feasibility of measuring their hypothesis at all given their population.

\subsubsection{Power}

Before we proceed, we will first give a broad overview of power in the two sample setting. In a fixed horizon setting, where valid confidence intervals are only available after a pre-determined amount of time, power calculations are a critical part of the experimentation process. So what is the statistical power? Power, usually denoted as $(1 - \beta)$, gives the expected proportion of times we would reject the null hypothesis when the alternative hypothesis is true. An intuitive way to understand $\beta$ is as a parameter that controls the type-II error, whereas the $\alpha$ we use in confidence intervals (and sequences) controls the type-I error.

\subsubsection{From Power to Expected Stopping Time}

Before we dive into the specifics of the power calculation, it is important to note that the power of an asymptotic confidence sequence is always $1$. This is due to the fact that, in contrast to confidence intervals resulting from fixed horizon experiments, AsympCS are defined with respect to an infinite horizon. In practical terms, this means that we are free to continually increase the number of subjects in an experiment until the null has been rejected, or we are sufficiently satisfied that we have failed to reject the null hypothesis.

Of course, in practice there aren't an infinite number of customers for us to add into an experiment, which means we still need some way of reasoning over the number of subjects we should expect to see before being able to reject the null under the alternative hypothesis. We will refer to this as the expected stopping time of an experiment.

More formally we will define the expected stopping time of an AsympCS as follows: Let $\hat{\mu}_n$ be the empirical mean at time $n$, and $\mathfrak{B}^{\alpha, \sigma}_n$ be the uncertainty interval, so that the $(1-\alpha)$ AsympCS for the true mean $\mu$ has the form 
\begin{align} 
\bar{C}_n \equiv\left(\widehat{\mu}_n \pm \mathfrak{B}^{\alpha,\sigma}_n\right):=\left(\widehat{\mu}_n \pm \widehat{\sigma}_n \sqrt{\frac{2\left(n \rho^2+1\right)}{n^2 \rho^2} \log \left(\frac{\sqrt{n \rho^2+1}}{\alpha}\right)}\right),
\end{align}
where $\alpha$ is the level of the test and $\hat{\sigma}_n$ is the sample standard deviation. In the case of A/B tests, we will swap symbols for the treatment effect and use $\hat{\theta}_n$ for the sample treatment effect, where the value under the null is $\theta_{H_0}$, while the value under the alternative hypothesis is $\theta_{H_1}$. 

We will assume $\theta_{H_0} = 0$ in the sequel without loss of generality. And we let $\theta_{H_1}$ be the treatment effect under the alternative hypothesis. This should be set to the desired minimal detectable effect. Finally, let $N$ be the total number of subjects that will be considered in the experiment (total population size).
For a given type-I and type-II error rate of $\alpha$, and $\beta$, respectively, the expected stopping time is defined in terms of the AsympCS under the null hypothesis, $\bar{C}^{H_0}_n$ as 
$$ \inf_n \left( p(\theta_{H_0} \not\in \bar{C}^{H_0}_{n} | H_1) >= 1 - \beta\right), $$ 
\textit{i.e}, we wish to find the earliest stopping time such that the probability that we reject the null hypothesis is at least $(1 - \beta)$. Note that the validity of this optimization problem stems directly from the definition of asymptotic confidence sequences. After noting that (1) $\sigma$ is the same under either hypothesis, and (2) the earliest possible time period is entailed by the time at which the $\beta$ quantile under $H_1$ is greater than or equal to the $(1 - \alpha)$ quantile under $H_0$, we can rewrite the expected stopping time problem as 
\begin{align}
\argmin_{n \in 1,\dots,N} \left(|\bar{C}^{H_0}_n| - (\theta_{H_1} - \mathfrak{B}^{\beta, \sigma}_t) >= 0\right),
\end{align}
in plain terms the optimization problem reduces to finding the time point at which the quantile entailed by our power constraint, $\beta$ exceeds the decision boundary given by our type-I error level, $\alpha$. Because both $\bar{C}^{H_0}_n$, and $(\theta_{H_1} - \mathfrak{B}^{\beta, \sigma}_n)$ are monotonic with respect to $n$, this optimization problem is, as a result, convex. While we are unaware of a closed form solution to this optimization problem, it can be efficiently solved with standard root finding procedures.

%%%%%%%%%%%%%%%%%%%%%%%%%%%%%%%%%%%%%%%%%%%%%%%%%%%%%%%%%%%%%%%%%%%%%%%%%%%%%%%%%%%%%%%%%%%%%%%%%%%
%%%%%%%%%%%%%%%%%%%%%%%%%%%%%%%%%%%%%%%%%%%%%%%%%%%%%%%%%%%%%%%%%%%%%%%%%%%%%%%%%%%%%%%%%%%%%%%%%%%
%%%%%%%%%%%%%%%%%%%%%%%%%%%%%%%%%%%%%%%%%%%%%%%%%%%%%%%%%%%%%%%%%%%%%%%%%%%%%%%%%%%%%%%%%%%%%%%%%%%

\subsection{AsympCS for Lift}\label{sec:lift}
Recall the two-sample hypothesis test that is common to most A/B tests. 
\begin{equation*}
H_0: \mu_0 = \mu_1 \quad vs. \quad H_1: \mu_0 \neq \mu_1, 
\end{equation*}
alternately, we can write it as 
\begin{equation}
\label{eqn:1}
H_0: \frac{\mu_1}{\mu_0} = 1 \quad vs. \quad H_1: \frac{\mu_1}{\mu_0} \neq 1. 
\end{equation}
In the Experimentation Service we use the Anytime Valid Confidence Sequence for the ATE as defined in \ref{eqn:asympcs} to declare a winner for this hypothesis test. From extensive simulations we see that this test has correct type I error and is anytime valid. 

Sometimes, marketers are interested in an alternate quantity, called the \emph{relative treatment effect}, more commonly referred to as the \emph{lift}, formally defined as $\theta = \mu_1/\mu_0 - 1$.
%We would like to construct a confidence interval for the below test statistic (in practice this is often referred to as the \emph{lift}, and presented in percentage terms):  
%\begin{equation*}
%    \theta = \frac{\mu_1}{\mu_0} - 1.
%\end{equation*}
We can then present the test statistic $\hat{\theta} = \hat{\mu}_1/\hat{\mu}_0 - 1$, along with a $(1-\alpha)$ confidence bound for $\theta$ to a user who is interested in Hypothesis (\ref{eqn:1}). We would like for this bound to be an Anytime Valid Confidence Sequence. 
\\

\noindent \textbf{Parameter change}. We define the parameter $\gamma = \log (\theta + 1) = \log \mu_1 - \log \mu_0$, \textit{i.e.}, $\theta = e^\gamma -1.$
Before we construct an Anytime Valid Confidence Sequence for $\theta$, let's look at some preliminary results we will use.  

\begin{lemma}[Monotonic Transformations]
\label{lem:1}
Let $X$ be a continuous random variable, such that: $\Pr(X < \ell) = \alpha/2$. Then, for a strictly monotonic increasing function $g: \mathbb{R} \rightarrow \mathbb{R}$, $\Pr(g(X) < g(\ell)) = \alpha/2$.
\end{lemma}
\begin{proof}
For the proof of this, note that a strictly monotonic function always has an inverse, $g^{-1}(\cdot)$ over the range of the function $\mathbb{R}$. Therefore, we have:
\begin{eqnarray}
\Pr(g(X) < g(\ell)) = \Pr(X < g^{-1}(g(\ell))) = \Pr(X < x) = \alpha/2.
\end{eqnarray}
\end{proof}
Our goal is to construct an AsympCS for $\gamma$, and then apply Lemma~\ref{lem:1} to get an AsympCS for $\theta$. %For a sequence of random variables $X_1, X_2, X_3, \cdots X_n$ with mean $\mu$, such that $2+\delta$ moments exist, for some positive $\delta$, Theorem 1 of \cite{waudby2021doubly} provides a two-sided AsympCS for the running mean: $$\hat{\mu}_n = \frac{1}{n} \cdot \sum_{i = 1}^n X_i.$$ 
We use the one-sided version of the AsympCS (Proposition C.1 of \cite{waudby2021doubly}), which states that for a sequence of random variables $X_{01}, X_{02}, X_{03}, \cdots$ with $\E(X_{0i}) = \mu_0$ and $\E(|X_{0i}|^{2+\delta}) < \infty$ (for some positive $\delta$), the running mean $\hat{\mu}_{0;n} = \frac{1}{n} \cdot \sum_{i=1}^{n_0} X_{0i}$ has the property:

\begin{equation}
\sup_n \Pr\left(\hat{\mu}_{0;n} - \hat{\sigma}_{0;n} \sqrt{\frac{2(n\rho^2 + 1)}{n^2 \rho^2} \log \left(\frac{\sqrt{(n \rho^2 + 1)}}{\alpha}\right)} > \mu_0 \right) \le \frac{\alpha}{2}, \mbox{ and }
\end{equation}
% \ian{The above is not true. This is only true for \emph{nonasymptotic} CSs, and the equality should be replaced with $\leq$. Also, should use $\backslash \mathrm{left} ( \text{math expression} \backslash \mathrm{right})$ to get brackets right.}

\begin{equation}
\sup_n \Pr\left(\hat{\mu}_{0;n} + \hat{\sigma}_{0;n} \sqrt{\frac{2(n\rho^2 + 1)}{n^2 \rho^2} \log \left(\frac{\sqrt{(n \rho^2 + 1)}}{\alpha}\right)} < \mu_0\right) \le \frac{\alpha}{2},
\end{equation}
where $\hat{\sigma}_{0_n}$ is the sample standard deviation.
%It is important to remember that the $n$ used to constructed the curved boundaries is the sample for each treatment arm, and not the sum for the two arms. 
For ease of notation, define the two terms:
\begin{eqnarray*}
\hat{l}_{0;n} &=& \hat{\mu}_{0;n} - \hat{\sigma}_{0;n} \sqrt{\frac{2(n\rho^2 + 1)}{n^2 \rho^2} \log \left(\frac{\sqrt{(n \rho^2 + 1)}}{\alpha}\right)}, \\
\hat{u}_{0;n} &=& \hat{\mu}_{0;n} + \hat{\sigma}_{0;n} \sqrt{\frac{2(n\rho^2 + 1)}{n^2 \rho^2} \log \left(\frac{\sqrt{(n \rho^2 + 1)}}{\alpha}\right)}.
\end{eqnarray*}
Applying Lemma~\ref{lem:1}, we get the following:
\begin{eqnarray*}
\sup_n \Pr(\log \hat{l}_{0;n} > \log \mu_0) &\le& \frac{\alpha}{2}\\
\sup_n \Pr(\log \hat{u}_{0;n} < \log \mu_0) &\le& \frac{\alpha}{2}.
\end{eqnarray*}

Applying the union bound, the set $(\log \hat{l}_{0;n}, \log \hat{u}_{0;n})$ is a $(1-\alpha)$ Anytime Valid Confidence Sequence for the treatment mean $\mu_0$. Define the AsympCS for $\mu_1$ similarly as $(\log \hat{l}_{1;n}, \log \hat{u}_{1;n})$. Next, we are interested in the difference $(\log \mu_1 - \log \mu_0) = \gamma$. To construct a confidence bound for $\gamma$, we take the maximum possible difference between the two test statistics within the ranges of the two AsympCSs above as the upper bound for $\gamma$. As the lower bound for $\gamma$ we take the smallest difference possible from the two. This is a standard, albeit conservative way of constructing a confidence bound for a different from the confidence bounds for the component pieces. Thus, the following is a $(1-\alpha)$ AsympCS for $\gamma$:
\begin{eqnarray}
\label{eqn:5}
(\log \hat{l}_{1;n} - \log \hat{u}_{0;n}, \log \hat{u}_{1;n} - \log \hat{l}_{0;n}).
\end{eqnarray}
Applying Lemma \ref{lem:1} again, we get a $(1-\alpha)$ AsympCS for $\theta$:
\begin{eqnarray}
(e^{\log \hat{l}_{1;n} - \log \hat{u}_{0;n}}-1, e^{\log \hat{u}_{1;n} - \log \hat{l}_{0;n}} -1).
\end{eqnarray}
Interestingly, this boils down to 
\begin{equation}
\left(\frac{\hat{l}_{1;n}}{\hat{u}_{0;n}} - 1, \frac{\hat{u}_{1;n}}{\hat{l}_{0;n}} - 1\right).
\end{equation}

%%%%%%%%%%%%%%%%%%%%%%%%%%%%%%%%%%%%%%%%%%%%%%%%%%%%%%%%%%%%%%%%%%%%%%%%%%%%%%%%%%%%%%%%%%%%%%%%%%%
%%%%%%%%%%%%%%%%%%%%%%%%%%%%%%%%%%%%%%%%%%%%%%%%%%%%%%%%%%%%%%%%%%%%%%%%%%%%%%%%%%%%%%%%%%%%%%%%%%%
%%%%%%%%%%%%%%%%%%%%%%%%%%%%%%%%%%%%%%%%%%%%%%%%%%%%%%%%%%%%%%%%%%%%%%%%%%%%%%%%%%%%%%%%%%%%%%%%%%%

\subsection{Baseline Approaches}
\label{sec:baselines}

In this section, we describe the baseline methods in more detail. 
\subsubsection{Fixed Horizon $z$-Tests}
The hypothesis test of interest is the ``difference of means'' between the two treatment arms. The hypothesis being tested for an A/B test is (conducted at the level $\alpha$):
\begin{equation}
H_0:\mu_0 = \mu_1     \quad\text{vs.}\quad   H_1: \mu_0 \neq \mu_1.
\label{eqn:hyp}
\end{equation}
The first method we consider  is a classical z-test (in large sample limit, this is equivalent to a t-test). Let $\hat{\theta}_n = \hat{\mu}_{1;n} - \hat{\mu}_{0;n}$ be the sample difference in means, and $\hat{\sigma}_n$ be the sample (pooled) standard deviation after $n$ samples have been recorded. For this test, the (point-wise) confidence interval is given by $CI_{1-\alpha} \coloneqq \hat{\theta}_n \pm z_{1-\alpha/2} \frac{\hat{\sigma}_n}{\sqrt{n}}$. Note that the statement that the null hypothesis is rejected (\textit{i.e.}, an effect size is non-zero) at significance level $\alpha$ is equivalent to a statement that the $(1-\alpha)$ CI does not include the null. Formally, the type-I error guarantees of this confidence interval are only satisfied if the result is looked at once. %this will be the first baseline approach and we will call it \textbf{(1) FHT} (fixed horizon test without peeking).

In what follows, we will also explore an alternate stopping rule, the test will be stopped as soon as the above $(1-\alpha)$ CI does not include the null hypothesis. Of course, this is formally an incorrect way to run a test. But it is emblematic of the peeking problem we are attempting to fix. This is our first baseline, \textbf{FHT-peeking} (fixed horizon test with peeking).

\subsubsection{Group Sequential methods, \textbf{LDM}}
As our second baseline, we consider the group sequential method define by \cite{gordon1983discrete, demets1994interim}, \textbf{LDM} (Lan and De Mets). This requires pre-specification of the total sample size, but allows flexible peeking (with the requirement that the times of the peeks are pre-registered and accounted for). We use the Pocock alpha-spending function \cite{pocock1977group}, which spends $\alpha$ evenly over the length of the experiment. The cumulative $\alpha$ spent by the peek at time $t^*$ is given by $\alpha \log(1+(e+1)t^*)$, with the incremental $\alpha$ spent at the $k^{\text{th}}$ peek given by $\alpha(t_k^*)-\alpha(t_{k-1}^*)$.  During numerical experiments, we use two sides boundaries with 100 equally spaced peeks, up until the fixed horizon sample size. 

\subsubsection{\textbf{mSPRT} method}

Our next baseline is the Mixture Sequential Probability Ratio Test proposed in \cite{johari2017peeking}. For binary data, they define a quantity
\begin{align}
\tilde{\Lambda}^{H, \theta_0}_{n} = \sqrt{\frac{\hat{\sigma}^2_n}{n \rho^2 + \hat{\sigma}^2_n }}\exp{\left\{\frac{n^2\rho^2(\hat{\theta}_n - \theta_0)^2}{2\hat{\sigma}^2_n(n \rho^2 + \hat{\sigma}^2_n )}\right\}}\nonumber
\end{align}
where $\theta_0$ is the assumed difference of means under the null hypothesis. The always valid $p$-values are defined by a recursion 
\begin{align}
p_0 = 1; \quad p_n = \min{\{p_{n-1}, 1/\Lambda^{H, \theta_0}_n\}}\nonumber
\end{align}
and the test is stopped as soon as $p \le \alpha$.

A corresponding $(1-\alpha)$-Confidence Sequence (termed an ``anytime valid confidence interval'' by \cite{johari2017peeking}), can be derived by inverting the mSPRT and is given by
\begin{equation}
\bar{C}_{n}^{\text{mSPRT}} \coloneqq \Biggl\{\hat{\theta}_n \pm \hat{\sigma}_n\sqrt{\frac{2(n\rho^2+\hat{\sigma}^2_n)}{n^2\rho^2} \log\left(\sqrt{\frac{n\rho^2+\hat{\sigma}^2_n}{\alpha \hat{\sigma}^{2}_n}}\right)}\Biggr\}.
\end{equation}
Recall, $\rho^2$ is a free parameter that can be tuned, and $\hat{\sigma}_n$ is the \textit{pooled sample variance}. The test is stopped as soon as the above $CS$ does not include the null hypothesis.

The stopping rule involves running the test and evaluating each new data point stopping whenever a this confidence sequence does not include the null, which is equivalent to saying $\tilde{\Lambda}^{H, \theta_0}_{n} >= 1/\alpha$. 

\subsubsection{Asymptotic Confidence Sequences - \textbf{AsympCS}}
The next method is of course Asymptotic Confidence Sequences, \textbf{AsympCS}. For this, we use the results of Appendix ~\ref{sec:appendix_ate}, to write the Conidence Sequence as 

\begin{align}
\bar{C}_{n}^{\text{AsympCS}} := \left\{ \hat{\theta}_n \pm  \pm \sqrt{\widehat{\text{var}}_{n}(\hat{f}_{n})} \cdot \sqrt{\frac{2(n\rho^2 + 1)}{n^2\rho^2}\log\left(\frac{\sqrt{n\rho^2 + 1}}{\alpha}\right)}\right\}
\end{align}
where the empirical inverse propensity weighted estimator is 
\begin{align}
\widehat{\text{var}}_{n}(\hat{f}_{n}) = \frac{n}{n-1}\left[ \frac{n}{n_0}\left(\hat{\sigma}^2_0 + \hat{\mu}^2_0\right) + \frac{n}{n_1}\left(\hat{\sigma}^2_1 + \hat{\mu}^2_1\right) - \left(\hat{\mu}_1 - \hat{\mu}_0\right)^2\right]
\end{align}
We note the deep similarities between this expression, and the \textbf{mSPRT} method's confidence sequence. 

Again, the stopping rule involves running the test, and with each new data point, evaluating whether the Confidence sequence contains the null. If it does not, the test is stopped.

\subsubsection{Bayesian Hypothesis Testing - \textbf{BHT-uninformed}}

A decision theoretic way of running A/B tests is to instead define a loss function $L(\theta, \theta^*)$ for an unknown parameter $\theta$, which quantifies the cost of deciding that the parameter has a different value $\theta^*$ instead of its ``true'' value, $\theta$. While decision theory spans both frequentist and Bayesian approaches, here we will focus on the latter. Bayesian methods will generally give a posterior distribution $p(\theta|X)$ for the unknown parameter $\theta$ after having observed data $X$. Decision theory then has us compute the posterior expected loss $E\left[L(\theta^*|X)\right] = \int L(\theta, \theta^*) p(\theta|X) d\theta$, which is a function of $\theta^*$. The optimal decision is then the choice of $\theta^*$ that minimizes this posterior expected loss. For A/B testing, there are of course many possible loss functions one can consider, e.g., the company VWO, in a white paper on Bayesian methods \cite{stucchio2015bayesian}, suggested a loss function that penalizes wrong choices linearly:
\begin{align}
L(\theta_0, \theta_1; 0) &= max(\theta_1 - \theta_0, 0)\nonumber\\
L(\theta_0, \theta_1; 1) &= max(\theta_0 - \theta_1, 0)\nonumber.
\end{align}
This should be read as saying ``the loss of choosing arm 0 in the test, is 0 if $\theta_0 > \theta_1$, but is equal to the amount of uplift lost if $\theta_1 > \theta_0$. The Expected Loss is then 

\begin{align}
\E[L](?) = \int^{1}_{0}d\theta_0 \int^{1}_{0}d\theta_1 L(\theta_0, \theta_1; ?)P(\theta_0, \theta_1),
\end{align}
For binary outcomes, we implement \cite{stucchio2015bayesian}, with a uniform Beta prior for propensity of response, so that the posterior for each (independent) arm is also a Beta distribution. Having observed $n_i$ units in arm $i$, with $c_i$ conversions, this takes the form  
\begin{align}
\Pr(\theta_0, \theta_1) &= \Pr(\theta_0)\Pr(\theta_1) \nonumber\\
&= \text{Beta}(\theta_0; \alpha + c_0, \beta + n_0 - c_0)\cdot\text{Beta}(\theta_1; \alpha + c_1, \beta + n_1 - c_1),
\end{align}
where we set $\alpha = \beta = 1$ for an ``uninformed'' Beta distribution, as is generally used in an industry A/B testing platform. We evaluate this posterior expected loss at every peek, and stop as soon as it is below a ``threshold of caring,'' $\varepsilon$. For a baseline conversion rate of $\theta_0 = 0.1$, we set $\varepsilon = 0.0001$

\subsubsection{Bayes Factors- \textbf{BF-uninformed}}
Another popular paradigm for Bayesian ``testing'' is the use of Bayes Factors \cite{deng2016continuous}. In this framework, we assume a prior probability for $H_i$ to be true ($p(H_i)$). After collecting data, $X$, for the experiment, we compute the posterior odds of BF: $p(H_1|X)/p(H_0|X)$. The BF is the ratio of the likelihood of observing data under $H_1$ relative to $H_0$. If we assumed both the hypotheses to be equally plausible to start with, the Prior Odds is $1$, and the Posterior Odds are equal to the Bayes Factor. We then run the test, and stop it when the Posterior Odds are above (or below) some threshold, e.g., above 10, (or below 0.1) representing 10:1 odds that the alternative (or null) is consistent with the data. In our implementation, for binary outcomes, we have the null and alternative hypotheses:
\begin{align}
H_0:\,\,\mu_0 = p_A\quad\text{and } \mu_1 = p_A \qquad H_1:\,\,\mu_0 = p_A\quad\text{and }\mu_1 = p_B
\end{align}
We can then define priors as $\Pr(p_A, p_B | H_1) = \Pr(p_A|H_1) \times \Pr(p_B|H_1)$, with individual probabilities following the Beta distribution: $\Pr(p_A|H_1) = \text{Beta}(p_A; \alpha_A, \beta_A)$, and 
$\Pr(p_B|H_1) = \text{Beta}(p_B; \alpha_B, \beta_B)$. Meanwhile, for the null hypothesis, we have $\Pr(p_A| H_0) = \text{Beta}(p_A; \alpha_C, \beta_C)$. 

With these definitions, and having observed $c_0$ conversions from $n_0$ units in the first arm, and $c_1$ conversions from $n_1$ units in the second arm, we have a Bayes factor that is expressed purely in terms of beta functions, $\beta$:
\begin{align}
BF &=  \frac{\beta(\alpha_A + c_0, \beta_A + n_0 - c_0)\beta(\alpha_B + c_1, \beta_B + n_1 - c_1)}{\beta(\alpha_A, \beta_A) \beta(\alpha_B, \beta_B)} \div \frac{\beta(\alpha_C + c_0 + c_1, \beta_C + n_0 + n_1 - c_0 - c_1)}{\beta(\alpha_C, \beta_C)}\nonumber\\
\end{align}
Setting $\alpha_A = \alpha_B = \alpha_C = \alpha$, and $\beta_A = \beta_B = \beta_C = \beta$, we get
\begin{align}
BF = \frac{\beta(\alpha + c_0, \beta + n_0 - c_0)\beta(\alpha + c_1, \beta + n_1 - c_1)}{\beta(\alpha, \beta) \beta(\alpha + c_0 + c_1, \beta + n_0 + n_1 - c_0 - c_1)}
\end{align}
We use uniform/uninformative priors, so we set $\alpha = \beta = 1$ throughout our simulations. 

We stop when the Bayes Factors is above $1/\alpha$ (this $\alpha$ is the type-I error threshold, and not the prior $\alpha$), and for $\alpha = 0.05$, this represents 20:1 odds that the alternative (or null) is consistent with the data. This is our final baseline, \textbf{BF-uninformed}.

%%%%%%%%%%%%%%%%%%%%%%%%%%%%%%%%%%%%%%%%%%%%%%%%%%%%%%%%%%%%%%%%%%%%%%%%%%%%%%%%%%%%%%%%%%%%%%%%%%%

\section{Additional Experimental Results}\label{sec:additional-experimental-results}
\subsection{Effect of Tuning Parameter $\rho^2$}\label{sec:tuning-parameter-rho}
AsympCS has a tuning parameter $\rho^2$, which controls the rate at which the type-I error is spent over the horizon of the test. While the anytime guarantees hold for all values of $\rho^2$, we explore the performance of AsympCS based on different values of the tuning parameter. Figures \ref{fig:tuning_parameter} and \ref{fig:tuning_parameter_power} shows the different values of the tuning parameter and its effect on the asymptotic type-1 error and power. Based on these we have selected a value of $10^{-3}$ as a value that ensures good properties over a range of possible scenarios. 

\begin{figure}[t]
\centering
\includegraphics[width = 0.9\linewidth]{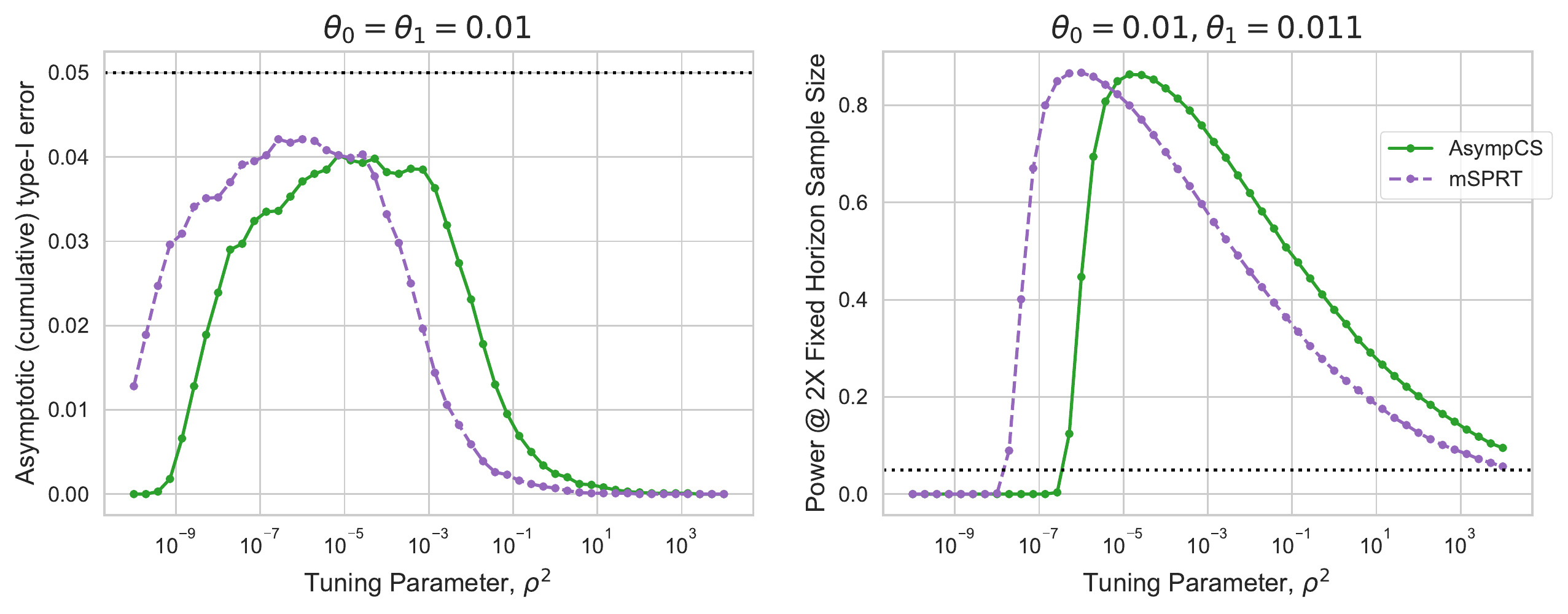}
\caption{Effect of the tuning parameter $\rho^2$ on asymptotic type-I error (left) and Power at 2x the fixed horizon stopping time (right) for the AsympCS and mSPRT methods . }
\label{fig:tuning_parameter}
\end{figure}
\begin{figure}[t]
\centering
\includegraphics[width = 0.5\linewidth]{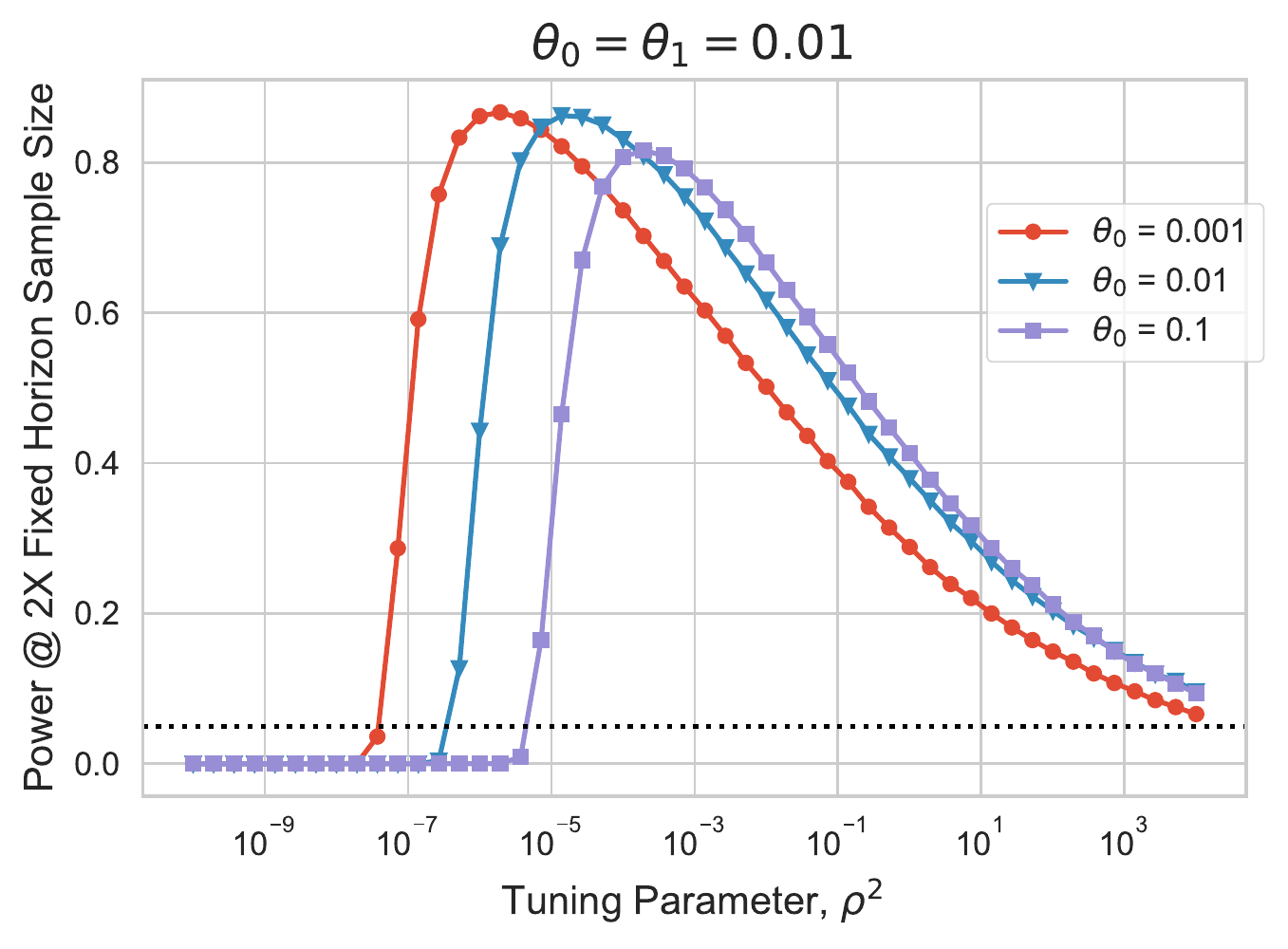}
\caption{Effect of the tuning parameter $\rho^2$ on Power (at 2x Fixed Horizon Sample Size) for different baseline conversion rates. We see that power dramatically falls below some critical value of $\rho^2$, so in practice, we have chosen a conservative $\rho^2 = 10^{-3}$ in production.}
\label{fig:tuning_parameter_power}
\end{figure}
\subsection{Empirical Distribution of Sample Size Ratios}\label{sec:sample-size-ratio}
Figure \ref{fig:sample_size_ratio} is our effort to recreate Figure 4 from \cite{johari2017peeking}. The construction of such a ``counterfactual'' graph based on empirical data is non-trivial: of course, we cannot run the real tests for longer than customers actually did. 

So, we instead construct this graph as follows: we take each of the $2089$ A/B tests in our sample, and extract the final observed difference in conversion rates (\textit{i.e.}, the final treatment effect $\hat{\mu_1} - \hat{\mu_0}$) for that experiment. We then use this as input to a sample size calculator with 80\% power and Type I error control of 5\%, and an assumed MDE that is exactly equal to $\hat{\mu_1} - \hat{\mu_0}$. The output of the sample size calculation is then a hypothetical ``perfectly calibrated fixed horizon sample size''. For each such real experiment, we then simulate $1000$ hypothetical experiments where the actual treatment effect is equal to this ``real one''. We then find the 80th percentile stopping time of these simulated experiments, and record the ratio of that 80th percentile stopping time to the Fixed Horizon stopping time at 80\% power. Figure \ref{fig:sample_size_ratio} plots the histograms of those ratios (green line). 

For the green line, the MDE that is used in the sample size calculator is exactly the observed treatment effect in the real experiment. We see that always valid methods are universally slower than Fixed Horizon methods, with runtimes between 1x and up to 4x the ``perfectly calibrated'' fixed horizon sample size. Of course, customers are rarely able to guess the true effect size in advance. For blue and orange lines, we assume that the user has underestimated the true effect, \textit{i.e.}, the MDE input to the sample size calculator is lower by 30\% or 50\%. This is the ``blockbuster'' effect scenario, \textit{i.e.}, the true effect is larger than the customer hoped for. In such cases, AsympCS method are often as fast, or faster than Fixed Horizon methods.

\subsection{Behavior of Bayesian Hypothesis Tests}\label{sec:bayesian}
As discussed in Section \ref{sec:t1error} of the main text, Bayesian methods are not designed to control Type-I error. In this section, we evaluate the Bayesian method against its stated goal - of bounding the expected loss.  The Bayesian methodology we have considered allows for a data dependent stopping rule, that compares expected loss under the posterior distributions to a ``threshold of caring'', $\varepsilon$. To show that this procedure works empirically, we adopted a Bayesian style simulation for a single armed analogy of an A/B test, where we simulate draws from a single arm, with the data distribution $\hat{\theta_0} \sim \text{Beta}(100,100)$ (which has its $95$\% CI from $\theta=0.44$ to $0.56$). This is compared against a fixed baseline conversion rate of $0.5$. The single arm setting is particularly amenable to simulations, since the loss function can be computed analytically.

For our modification of the A/B test, we will adapt their methods to the following probability model:

\begin{align}
\text{prior,}\quad\theta &\sim \text{Beta}(a_0, b_0)\\
\text{observations,}\quad x_i &\sim \text{Bernoulli}(\theta)\\
\text{likelihood,}\quad x_{1:n} &\sim \text{Binomial}(n,\theta)
\end{align}
where we will tune the free parameters of the prior $a_0$ and $b_0$. We will plot credible intervals for data drawn from this distribution, \textit{i.e.}, after having observed data  $x_{1:n}$, where there are $c$ successes, from $n$ samples, we have the posterior distribution

$$
\theta | x_{1:n} \sim \text{Beta}(a_0 + c, b_0 + n - c)
$$

The 95\% \emph{Credible Interval} is given by computing the 5\% and 95\% tails of this distribution. 
The stopping rule after having observed $n$ samples (\textit{i.e.}, with data $x_{1:n}$) is based on comparing the expected loss to a threshold
\begin{align}
\E[\max(\theta_0 - \theta, 0)| x_{1:n}] < \tau \quad \text{if } \hat{\theta} > \theta_0\\
\E[\max(\theta - \theta_0, 0)| x_{1:n}] < \tau \quad \text{if } \hat{\theta} \le \theta_0
\end{align}
where $\varepsilon$ is some pre-specified threshold of caring, and $\theta_0$ is the baseline conversion rate against which we are making the comparison. If either of these inequality tuplets is satisfied, the test is stopped. 

As mentioned previously, it turns out that this loss function can be computed analytically for a test with just one arm being compared to a baseline. For a posterior distribution with parameters $\alpha$ and $\beta$, the first half of the loss function has the form:

\begin{align}
\E[\max(\theta_0 - \theta, 0)| x_{1:n}] &= \int^{1}_{0}  \max(\theta_0 - \theta, 0)\frac{\theta^{\alpha-1}(1-\theta)^{\beta -1}}{B(\alpha, \beta)}d\theta\\
&= \int^{\theta_0}_{0}(\theta_0 - \theta)\frac{\theta^{\alpha-1}(1-\theta)^{\beta -1}}{B(\alpha, \beta)}d\theta  \\
&= \theta_0\int^{\theta_0}_{0} \frac{\theta^{\alpha-1}(1-\theta)^{\beta -1}}{B(\alpha, \beta)}d\theta - \frac{B(\alpha + 1, \beta)}{B(\alpha, \beta)}\int^{\theta_0}_{0} \frac{\theta^{(\alpha + 1) - 1}(1-\theta)^{\beta -1}}{B(\alpha + 1, \beta)}d\theta\\
&= \theta_0 I_{\theta_0}(\alpha, \beta) - \frac{\alpha}{\alpha + \beta} I_{\theta_0}(\alpha + 1, \beta)
\end{align}
where $I_{x}(a,b)$ is the *regularized incomplete Beta function*. Similarly for the other loss function we have

$$
\E[\max(\theta - \theta_0, 0)| x_{1:n}] = \frac{\alpha}{\alpha + \beta}I_{1-\theta_0}(\beta, \alpha + 1) - \theta_0I_{1-\theta_0}(\beta,\alpha)
$$

In Figure \ref{fig:bayesianloss}, we show the empirical/observed loss at the stopping time in 5,000 simulations with 1000 peeks spaced logarithmically equally in the range $100$ to $10^8$ visitors. The plots shows two different Bayesian models - one where the prior is perfect, and accurately captures the data distribution, and one where the prior is an uninformative or ``flat'' prior of Beta(1,1). With the Bayesian decision rule that allows for data dependent stopping, we see that the empirical expected loss is exactly our threshold of caring when the prior is accurate. When the prior is mis-specified, the empirical loss exceeds the threshold that was set.

\begin{figure}
\centering
\includegraphics[width=0.5\linewidth]{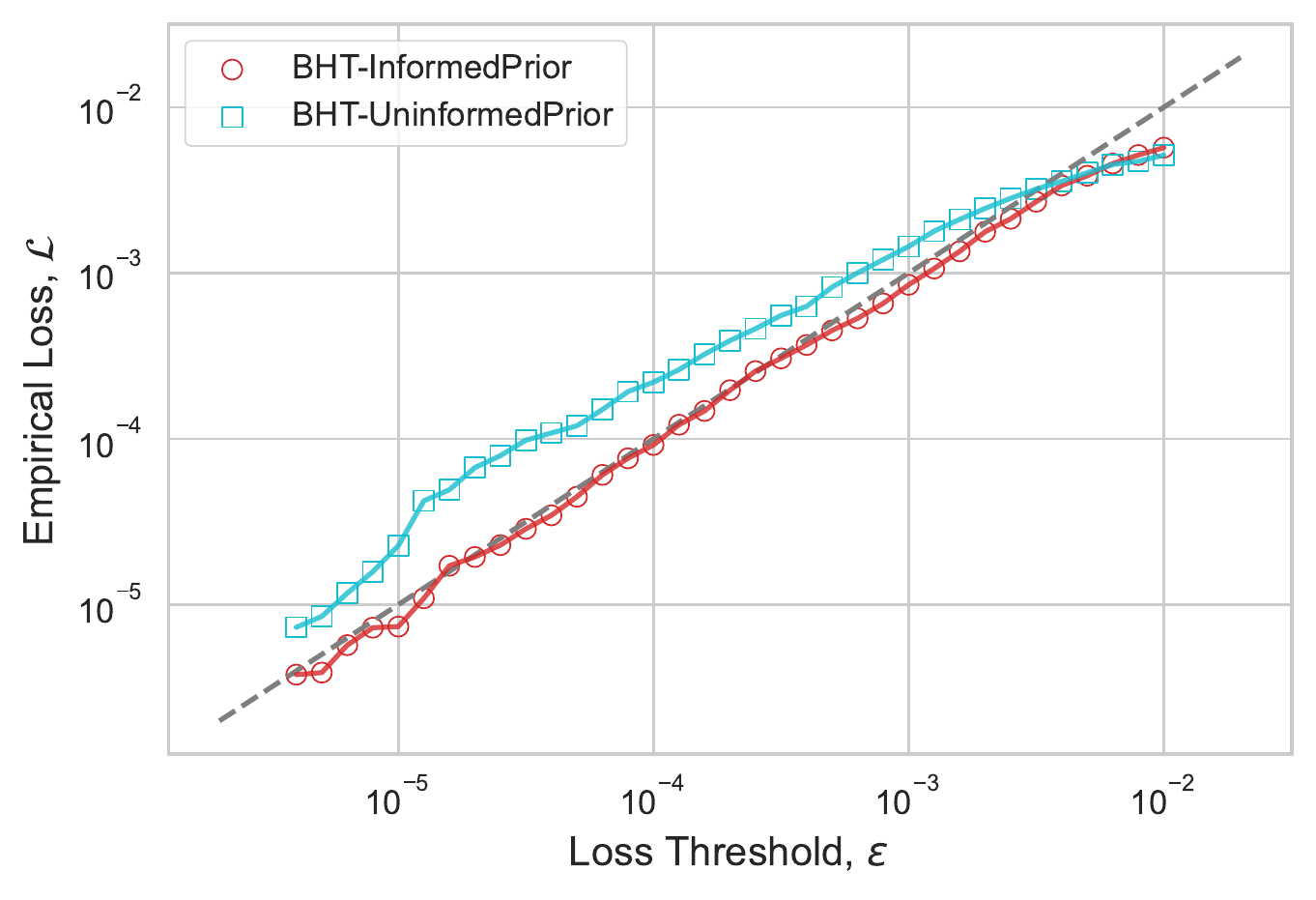}
\caption{Empirical Loss of Bayesian method  vs. threshold of caring in a simulated single armed test where the conversion rate is drawn from a Beta(100, 100) distribution. The blue points are empirical losses with a prior that is equal to the true data distribution, while the brown points are for a flat ``uninformative'' prior. When the prior accurately captures the data distribution, the empirical loss is exactly as we would expect.}
\label{fig:bayesianloss}
\end{figure}

\subsection{Comparisons of Anytime-valid methods and Bayesian Hypothesis Tests}\label{sec:comparison-anytime-bayesian}
In the main text, we showed a few frequentist properties of Bayesian methods. In this section, we provide a few more detailed comparisons of anytime valid methods and Bayesian Hypothesis tests, including several ``Bayesian style" comparisons. Note, we do not include Bayes Factor methods, since these have been covered extensively by \citet{deng2016continuous}.

\subsubsection{Miscoverage probability at stopping time}
\label{sec:miscoverage}
One very natural question is to examine miscoverage probabilities \textit{at the stopping time} for a given methodology. While AsympCS and mSPRT methods stop when their confidence sequences exclude the null, the Bayesian hypothesis test stops when its loss function drops below a threshold of caring $\varepsilon$. 

\begin{figure}
\centering
\includegraphics[width=0.6\linewidth]{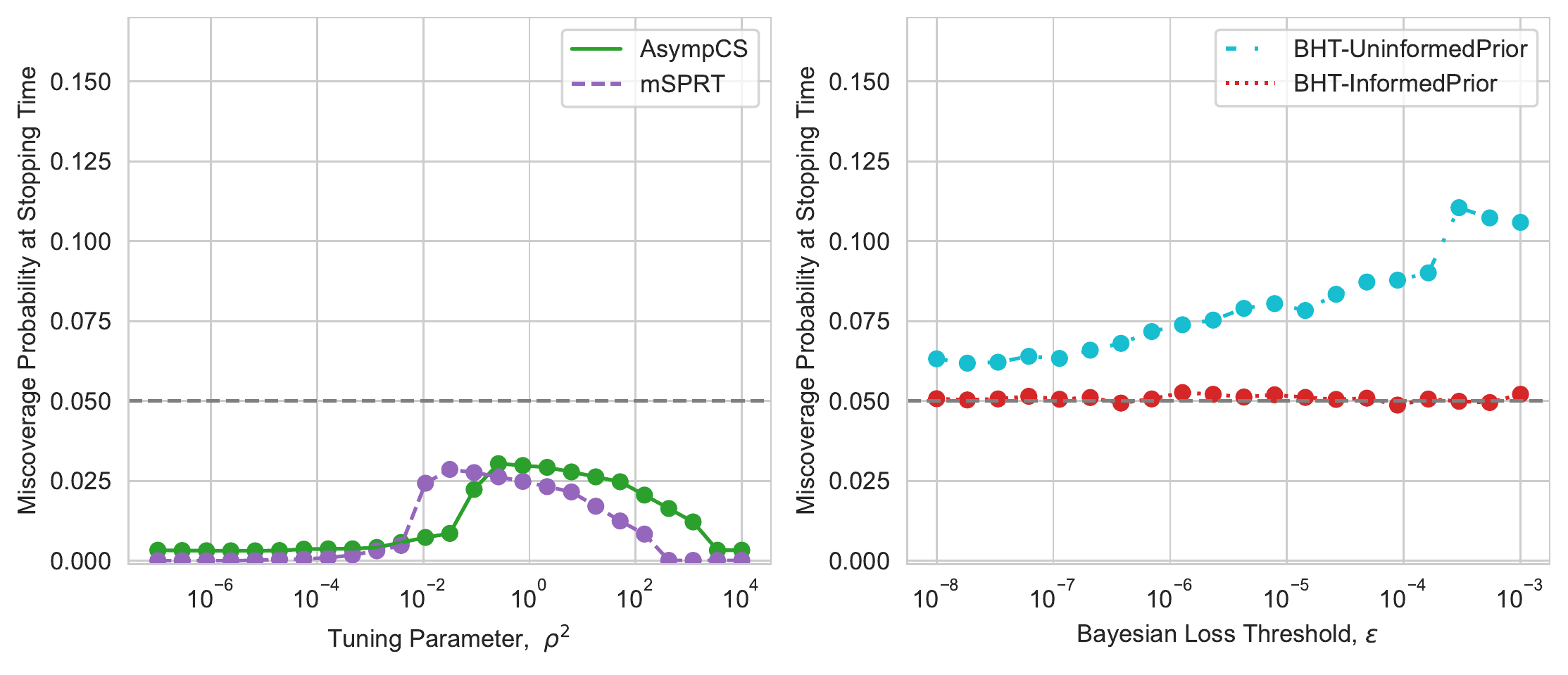}
\caption{Miscoverage probability \textit{at the stopping time} for $95\%$ Sequential Confidence Sequences (top) and $95\%$ Bayesian Credible Intervals (lower), as the respective tuning parameter of each methodology is varied. For Sequential methods, miscoverage error at stopping time is strictly below the $5\%$ bound, while for Bayesian methods, any mis-specification of the prior results in violation of bounds}
\label{fig:miscoverageatstopping}
\end{figure}

Figure \ref{fig:miscoverageatstopping} presents the variations in the miscoverage error at the stopping time as a function of each method's tuning parameters, for $10^4$ simulations with data drawn from a Beta(100,100) prior. The setting here is the ``single-arm'' test, where we compare a single arm's conversion rate to a baseline of $0.5$. A \textit{miscoverage} is defined as a scenario where the 95\% credible interval (CI) / confidence sequence (CS) at the stopping time, does not include the true parameter. 

The Sequential methods have strong worst case guarantees on their miscoverage error (\textit{i.e.}, it is always below $5\%$), while Bayesian CIs are well calibrated if the data matches the prior (exactly $5\%$ error). However, if the prior is inappropriate, this miscoverage error is not controlled. The general lesson here is that the anytime-valid methods allow for stopping at any time, and the inference at those times is unambiguous, with strong coverage guarantees. On the other hand, Bayesian methods provide coverage guarantees \textit{at their stopping time, and only if the prior is appropriate}.

\subsubsection{Power and Stopping time}
The next natural question is to compare the power and stopping time of these methods. This is shown in Figure~\ref{fig:powerstoppingtimecomparison}, where we see that Bayesian methods stop slightly faster than anytime valid methods. The simulation setup is identical to Section \ref{sec:miscoverage}. 

We note however that Bayesian methods with misspecified priors can never achieve the same miscoverage guarantees as Anytime valid methods. In fact, the tightest miscoverage guarantees come at very low loss thresholds (seen in Figure ~\ref{fig:miscoverageatstopping}), and for these thresholds, the 80th percentile stopping times are similar to Anytime Valid methods (top panel of Figure \ref{fig:powerstoppingtimecomparison}).  

\begin{figure}
\centering
\includegraphics[width=0.7\linewidth]{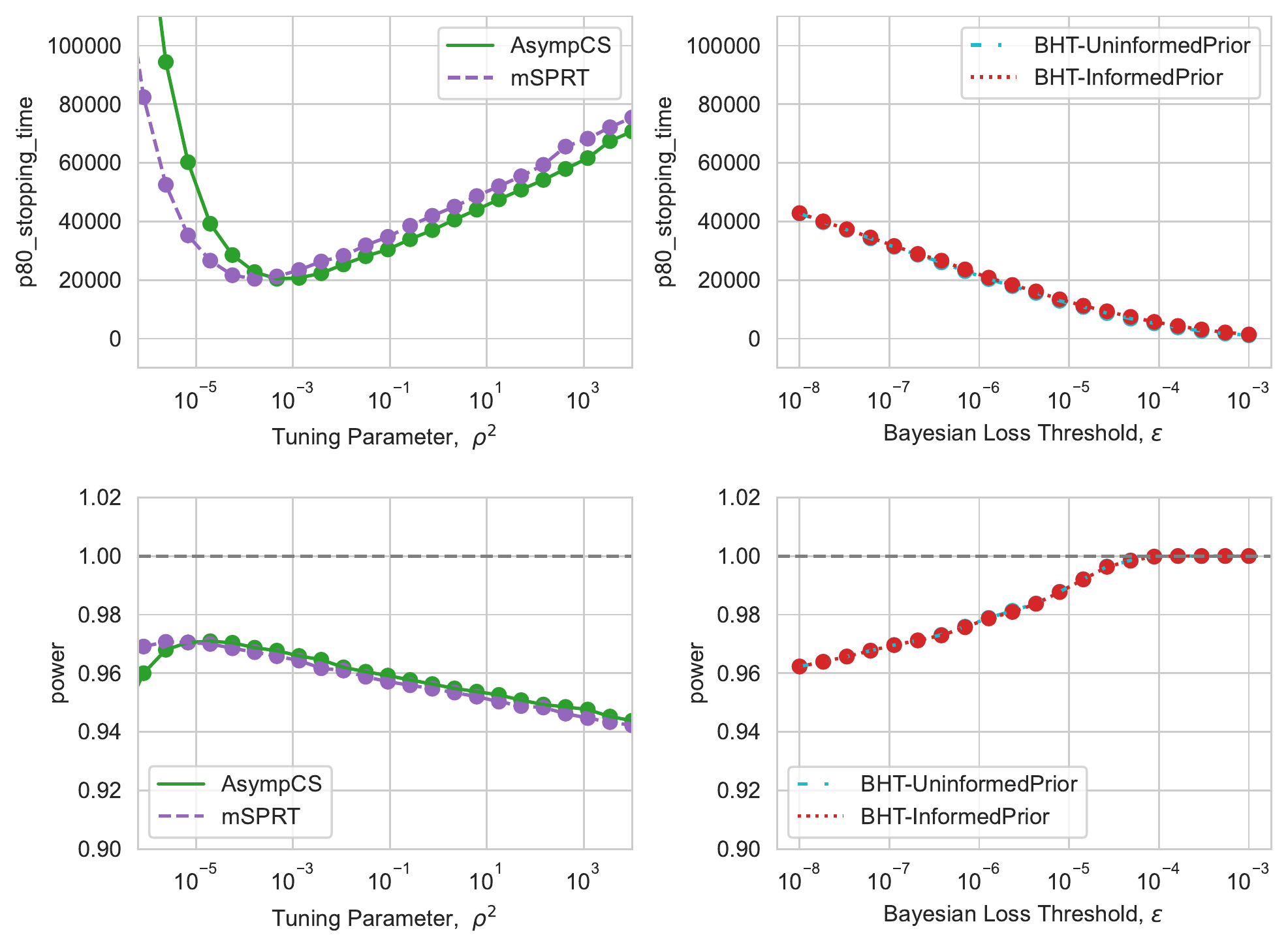}
\caption{A comparison of the 80th percentile stopping time (top) and Power (lower) for Anytime Valid methods vs. Bayesian Hypothesis tests. We see that Bayesian methods stop faster (in general), and as a result, have slightly higher power.}
\label{fig:powerstoppingtimecomparison}
\end{figure}

\subsubsection{Calibration of methods}
The final comparison we make is a calibration plot - comparing the ``inferred" vs. true effect sizes at the stopping time for each method. We again simulated 10,000 experiments, with effect sizes drawn from a Beta(100, 100) prior, and ran each test for each methodology until the stopping criterion was satisfied. The observed effect size at the stopping time is the ``inferred" effect, and this can be plotted against the true effect for that simulation run. The results of each of these 10,000 plots in shown in the scatter plots on the left column of Figure ~\ref{fig:calibration}. 

To interpret this data, we take both vertical and horizontal ``binned" slices. The middle column of Figure ~\ref{fig:calibration} shows a forward inferencing viewpoint - given the true effect size ($x$-axis), what effect size did we actually observe? We see that the anytime valid methods (AsympCS and mSPRT) are perfectly calibrated from this viewpoint. 

Next, from a backwards inferencing standpoint, \textit{i.e.}, \emph{given the data}, how does our inferred effect correspond to the parameter values that could have produced this data? The third column of Figure \ref{fig:calibration} shows these results. We see that only a BHT with an informed prior (that exactly matches the data generation process) is perfectly calibrated in this scenario. All other methods have some bias in this analysis, though anytime valid methods are better than other methods. Importantly, Bayesian methods with uninformed priors are poorly calibrated in both directions. 

\begin{figure}
\centering
\includegraphics[width=0.65\linewidth]{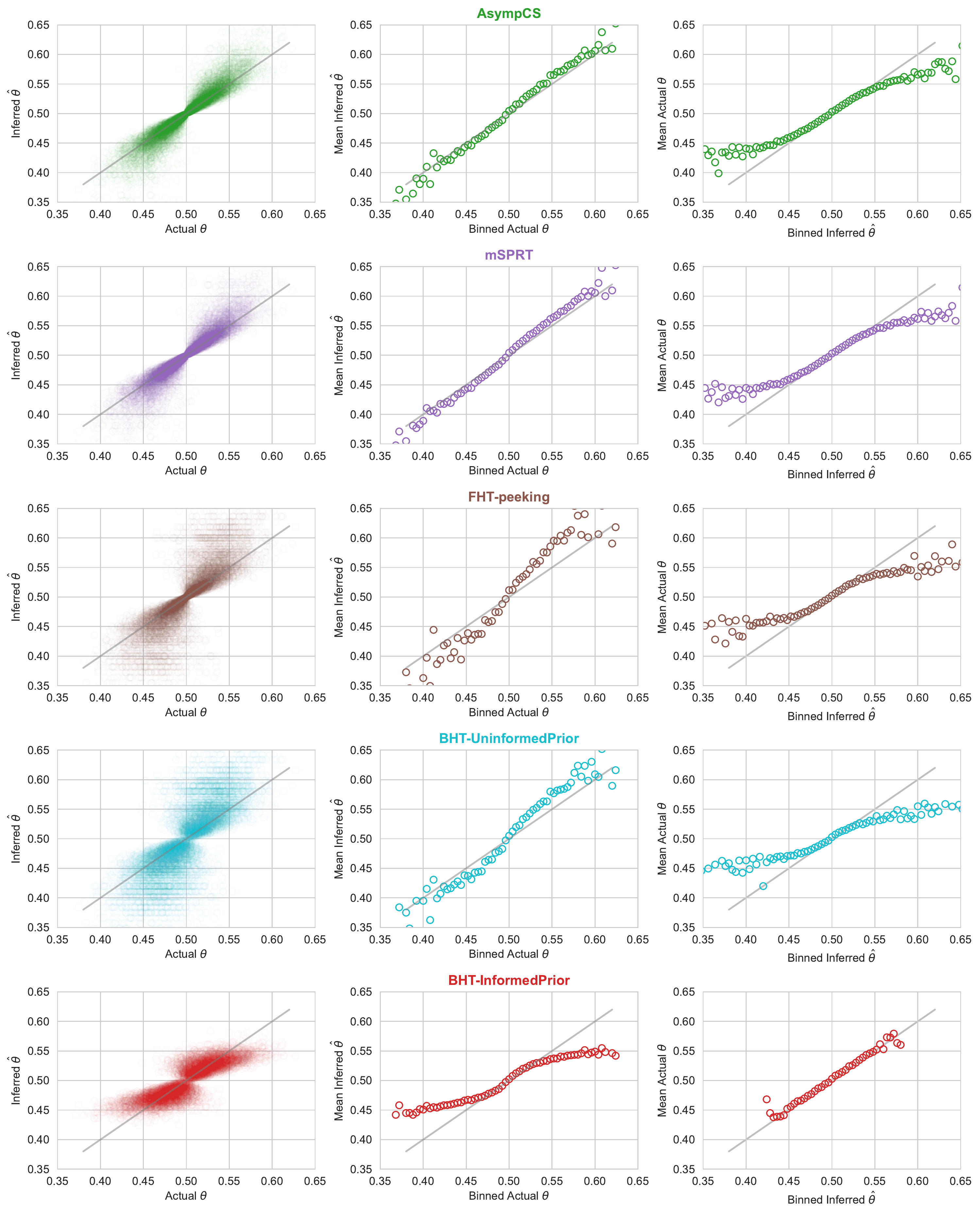}
\caption{Calibration Plots for various statistical methods in a single arm test (Left Column) Scatter Plots of Inferred/Posterior effect sizes vs. true effect sizes for 10,000 experiments each with a true effect drawn from a Beta(100, 100) distribution, and stopped according to the appropriate stopping rule of the method. (Middle Column) Calibration plots of inferred effect size vs. true effect size. (Right Column) Calibration plots of actual effect size given the data (i.e. given the inferred effect size)}
\label{fig:calibration}
\vspace{-10mm}
\end{figure}

\subsubsection{Conclusions}
The general lesson here is that the anytime-valid methods allow for stopping at any time, and the inference at those times is unambiguous, with strong coverage guarantees. On the other hand, Bayesian methods provide coverage guarantees \textit{at their stopping time \textbf{only} if the prior is appropriate}.

\end{document}